\documentclass[manuscript]{acmart}
\AtBeginDocument{%
  }

\setcopyright{acmlicensed}
\copyrightyear{2018}
\acmYear{2018}
\acmDOI{XXXXXXX.XXXXXXX}

\usepackage{xcolor}
\usepackage{xspace}
\usepackage{cleveref}
\usepackage{centernot}
\usepackage[inline]{enumitem}

\usepackage{thm-restate}

\usepackage{pgfplots}
\usepgfplotslibrary{dateplot}

\usepackage{tikz}
\usetikzlibrary{decorations.pathreplacing,calc,shapes,positioning}

\definecolor{kirby}{RGB}{215,72,148}
\definecolor{kirby-light}{RGB}{242,198,221}
\definecolor{kirby-blue}{RGB}{68,102,191}
\definecolor{pzgreen}{rgb}{0,0.50,0.16}

\usepackage{adjustbox}
\usepackage[backgroundcolor=lightgray]{todonotes}

\newtheorem{remark}{Remark}

\setcopyright{none}

\usepackage{amsmath}

\usepackage{booktabs}   
\usepackage{subcaption} 
                        
\usepackage{listings}
\usepackage{parcolumns}
\usepackage[utf8]{inputenc}

\usepackage{mathpartir}
\newcommand{\Coloneqq}{\mathrel{\mathop{::}}=}
\newcommand{\metaDef}{\mathrel{\mathop:}=}

\newcommand{\seq}{\ ; \ }

\newcommand{\disj}{\!\mathrel{\disjoint}\!}
\usepackage[font=small,skip=0pt]{caption}
\usepackage{tikz-cd}
\usepackage{graphicx}
\makeatletter
\let\@authorsaddresses\@empty
\makeatother
\makeatletter
\renewcommand\@formatdoi[1]{\ignorespaces}
\makeatother
\renewcommand\footnotetextcopyrightpermission[1]{}
\usepackage{chor}
\usepackage{symbols}

\def\apnof#1{\mathsf{apn}(#1)}

\def\bllabel{branch-selection label\xspace}
\def\bllabels{branch-selection labels\xspace}
\def\mbrel{relavance\xspace}
\def\mbGtoL{global to local\xspace}
\def\mbLtoG{local to global\xspace}

\def\defn#1{{\emph{#1}}}
\def\rclr#1{{\color{red}#1}}
\def\bclr#1{{\color{blue}#1}}

\def\gclr#1{{\color{pzgreen}#1}}

\def\quand{\quad\text{and}\quad}
\def\qomma{\;,\;}

\def\mydot{\quad .}
\def\resp#1{(resp. #1)\xspace}

\def\Glab{\mathbf{GLab}}
\def\Llab{\mathbf{LLab}}

\def\set#1{\{#1\}}

\def\setdef#1#2{\left\{ \begin{array}{l@{\;\vrule\;}l} #1 & #2 \end{array} \right\}}
\def\dom#1{\mathsf{dom}(#1)}
\def\tuple#1{\langle  #1  \rangle}
\def\pair#1#2{\tuple{#1 \;,\; #2}}

\def\relSymb{\mathrel{\mathcal{R}}}

\newcommand{\mbb}[1][]{\overset{#1}{\sqsubseteq}}
\newcommand{\mbbn}[2][n]{\mbb[#2]_{\pid {#1}}}
\def\bisim{\sim}
\newcommand{\bisimn}[1][n]{\sim_{#1}}

\newcommand{\procDef}[3][]{#2_{#1} \triangleq #3_{#1}}
\newcommand{\genProcDef}[1][]{\procDef[#1]{\cX}{\cC}}

\newcommand{\genpProcDef}[1][]{\procDef[#1]{\pX}{\pP}}

\def\chorpoly#1{\bclr{\mathsf{#1}}}
\newcommand{\cC}[1][]{\chorpoly{C_{#1}}}
\newcommand{\cCp}[1][]{\chorpoly{C^{\scriptscriptstyle\prime}_{#1}}}
\newcommand{\cCpp}[1][]{\chorpoly{C^{\scriptscriptstyle\prime\prime}_{#1}}}
\newcommand{\cX}[1][]{\chorpoly{X_{#1}}}
\newcommand{\cI}[1][]{\chorpoly{{I_{#1}}}}
\def\cnil{\chorpoly{\mathbf{0}}}

\newcommand{\procC}{\mathcal{C}}
\def\Chor{\mathcal{C}}

\def\lleft{\mathcal{l}}
\def\lright{\mathcal{r}}

\def\exppoly#1{\textit{#1}}
\def\ee{\exppoly{e}}

\def\cifte#1#2#3#4{\chorpoly{\textsf{if}\ \pid{#1}.{\exppoly{#2}} \ \textsf{then}\ #3\ \textsf{else}\ #4}}

\def\cloc#1#2#3{\chorpoly{\pid{#1}.{#2} \metaDef \exppoly{#3} }}

\newcommand{\cccom}[5][\pid]{\chorpoly{\ccom[#1]{#2}{\exppoly{#3}}{#4}{#5}}}
\newcommand{\ccsel}[4][\pid]{\chorpoly{\csel[#1]{#2}{#3}{#4}}}

\def\cgenloc{\cloc pxe}
\def\cgencom{\chorpoly{\ccom p\ee qx}}
\def\cgensel{\chorpoly{\gensel}}
\def\cgenifte{\cifte pe{\cC_1}{\cC_2}}
\def\cgeniftep{\cifte pe{\cC'_1}{\cC'_2}}

\def\pnof#1{\pn(#1)}

\def\Pset{\mathbf{PN}}
\def\Varset{\mathbf{V}}
\def\Valset{\mathbf{Values}}
\def\Eset{\mathbf{E}}

\def\CXset{\mathbf{CV}}
\def\PXset{\mathbf{PV}}
\def\ChorSet{\mathbf{Chor}}
\def\NetSet{\mathbf{Net}}
\def\Lset{\mathbf{L}}
\def\ELLset{\mathbf{\mathcal{L}}}
\def\lleft{\lbl l}
\def\lright{\lbl r}

\def\tradsem#1{\xrightarrow[]{#1}}
\newcommand{\localsem}[2][r]{\xrightarrow[]{#2}_{\pid{#1}}}

\newcommand{\notlocalsem}[2][r]{\not\!\xrightarrow[]{#2}_{\pid{#1}}}

\newcommand{\aggsem}[1]{\xrightarrow{#1}_{\mathcal{A}}}

\def\statefulsem#1{\xrightarrow[]{#1}}

\newcommand{\TRassign}[3][r]{\pid{#1}.#2 \metaDef {\exppoly{#3}}}
\newcommand{\TRcom}[4][e]{\pid{#2}.{\exppoly #1} \tto \pid{#3}.#4}
\newcommand{\TRsel}[3][\albl]{\csel{\pid{#2}}{\pid{#3}}{#1}}
\newcommand{\TRthen}[2][p]{\tlthen{\exppoly #2} @ \pid{#1}}
\newcommand{\TRelse}[2][p]{\tlelse{\exppoly #2} @ \pid{#1}}

\newcommand{\PRassign}[3][r]{\pid{\rclr{#1}}.#2 \metaDef {\exppoly{#3}}}
\newcommand{\PRsend}[3][e]{\pid{\rclr{!#2}}.{\exppoly{#1}} \tto \pid{#3}}
\newcommand{\PRrecv}[3][x]{\pid{{#2}} \tto \pid{\rclr{?#3}}.#1}
\newcommand{\PRlsend}[3][\albl]{\csel{\rclr{\pid{!#2}}}{\pid{#3}}{#1}}
\newcommand{\PRlrecv}[3][\albl]{\csel{\pid{#2}}{\pid{\rclr{?#3}}}{#1}}
\newcommand{\PRthen}[2][r]{\TRthen[\rclr{\pid{#1}}]{#2}}
\newcommand{\PRelse}[2][r]{\TRelse[\rclr{\pid{#1}}]{#2}}

\newcommand{\delaypredictl}[1][r]{{\small \scriptsize DELAY-PREDICT-LEFT}@$\pid{#1}$\xspace}
\newcommand{\delaypredictr}[1][r]{{\small \scriptsize DELAY-PREDICT-RIGHT}@$\pid{#1}$\xspace}
\newcommand{\delaycond}[1][r]{{\small \scriptsize DELAY-COND}@$\pid{#1}$\xspace}
\newcommand{\csend}[1][r]{{\small \scriptsize SEND}@$\pid{#1}$\xspace}
\newcommand{\crec}[1][r]{{\small \scriptsize REC}@$\pid{#1}$\xspace}
\newcommand{\csendl}[1][r]{{\small \scriptsize SENDL}@$\pid{#1}$\xspace}
\newcommand{\crecl}[1][r]{{\small \scriptsize RECL}@$\pid{#1}$\xspace}
\newcommand{\cdelay}[1][r]{{\small \scriptsize DELAY}@$\pid{#1}$\xspace}
\newcommand{\ccall}[1][r]{{\small \scriptsize CALL}@$\pid{#1}$\xspace}
\newcommand{\cthen}[1][r]{{\small \scriptsize COND-THEN}@$\pid{#1}$\xspace}
\newcommand{\celse}[1][r]{{\small \scriptsize COND-ELSE}@$\pid{#1}$\xspace}
\newcommand{\ccassign}[1][r]{{\small \scriptsize ASSIGN}@$\pid{#1}$\xspace}

\newcommand{\tassign}{{\small \scriptsize ASSIGN}\xspace}
\newcommand{\tcom}{{\small \scriptsize COM}\xspace}
\newcommand{\tsel}{{\small \scriptsize SEL}\xspace}
\newcommand{\tdelay}{{\small \scriptsize DELAY}\xspace}
\newcommand{\tdelaycond}{{\small \scriptsize DELAY-COND}\xspace}

\newcommand{\tcall}{{\small \scriptsize CALL}\xspace}
\newcommand{\tthen}{{\small \scriptsize COND-THEN}\xspace}

\newcommand{\parleft}{{\small \scriptsize PAR-LEFT}\xspace}
\newcommand{\parright}{{\small \scriptsize PAR-RIGHT}\xspace}

\definecolor{pzbrickred}{rgb}{0.5,0.0,0.10}

\def\netpoly#1{{\color{pzbrickred}\mathsf{#1}}}
\newcommand{\nN}[1][]{\netpoly{N_{#1}}}
\newcommand{\nNof}[1]{\procpoly{N(#1)}}
\newcommand{\nM}[1][]{\netpoly{M_{#1}}}
\newcommand{\nNp}[1][]{\netpoly{N^{\scriptscriptstyle\prime}_{#1}}}
\newcommand{\nNpp}[1][]{\netpoly{N^{\scriptscriptstyle\prime\prime}_{#1}}}

\def\nnil{\netpoly{\mathbf{0}}}

\def\procpoly#1{\gclr{\mathsf{#1}}}
\newcommand\pP[1][]{\procpoly{P_{#1}}}
\newcommand\pPp[1][]{\procpoly{P_{#1}^{\scriptscriptstyle\prime}}}
\newcommand\pQ[1][]{\procpoly{Q_{#1}}}
\newcommand\pR[1][]{\procpoly{R_{#1}}}
\newcommand\pQp[1][]{\procpoly{Q_{#1}^{\scriptscriptstyle\prime}}}
\newcommand{\pX}[1][X]{\procpoly{#1}}

\def\pseq{\seq}
\newcommand{\pdo}[2][p]{\netpoly{\pid{#1}\!\left[ {\color{black} #2} \right]}}
\def\ppar{\;\netpoly{|}\;}
\def\pnil{\procpoly{\mathbf{0}}}
\def\psend#1#2{\procpoly{\pid{#1}!{\exppoly{#2}}}}
\def\precv#1#2{\procpoly{\pid{#1}?{#2}}}
\newcommand{\plab}[3][\albl]{\{#1 : #3\}_{#1 \in #2}}
\newcommand{\psendl}[2][\albl]{\procpoly{\pid{#2}\oplus #1}}
\newcommand{\precl}[2][\albl]{\procpoly{\pid{#2}\& #1}}

\def\pgenass{\procpoly{x \metaDef \exppoly{e}}}
\def\palpha{\procpoly{\alpha}}

\def\pifte#1#2#3{\procpoly{\textsf{if} \ \exppoly{#1} \ \textsf{then} \ #2\ \textsf{else}\ #3}}
\def\genpifte{\pifte e {\pP_1}{\pP_2}}

\newcommand{\ponelabel}[3][\pP]{\procpoly{{\pid  #2} \& \{#3 : #1\}}}
\newcommand{\plabel}[4][\pP]{\procpoly{{\pid  #2} \& \plab[#3]{#4}{#1_#3}}}
\newcommand{\palabel}[4][\pP]{\procpoly{{\pid  #2} \& \plab[#3]{#4}{#1}}}
\newcommand{\plabels}[2]{\procpoly{{\pid #1}\& \set{#2}}}
\def\genplabel{\plabel{p}{\albl}{\ELLset}}

\newcommand{\proctr}[1]{\xrightarrow{#1}}
\newcommand{\netsem}[1]{\xrightarrow{#1}}
\newcommand{\conetsem}[1]{\xleftarrow{#1}}
\newcommand{\notnetsem}[1]{\not\!\xrightarrow{#1}}
\newcommand{\supp}{\mathsf{supp}}

\newcommand{\procP}{\mathcal{P}}

\newcommand{\proj}[2][]{%
  \if\relax\detokenize{#1}\relax%
    \netpoly{\left\llbracket{\color{black}#2}\right\rrbracket}
  \else%
    \procpoly{\left\llbracket{\color{black}#2}\right\rrbracket_{\color{black}\pid{#1}}}%
  \fi%
}

\newcommand{\creds}{\exppoly{creds}}
\newcommand{\valid}{\exppoly{valid}}
\newcommand{\token}{\exppoly{token}}
\newcommand{\okval}{\exppoly{ok}}
\newcommand{\koval}{\exppoly{fail}}
\newcommand{\resvar}{{res}}
\newcommand{\logger}{\pid{log}}

\def\storepoly#1{{\color{magenta}#1}}
\def\lupdate#1#2{\storepoly{[#1\mapsto #2]}}
\def\update#1#2#3{\storepoly{[#1.#2 \mapsto #3]}}
\newcommand{\sstore}[1][]{\storepoly{\sigma_{#1}}}
\newcommand{\sStore}[1][]{\storepoly{\Sigma_{#1}}}
\newcommand{\sStorep}[1][]{\storepoly{\Sigma'_{#1}}}
\def\evalin#1#2#3{#1 \vdash {\exppoly{#2}} \downarrow #3}

\usepackage[normalem]{ulem}

\begin{document}

\title{Choreographic Programming: a Semantic Approach}

\author{Matteo Acclavio}
\author{Giulia Manara}
\author{Fabrizio Montesi}
\author{Xueying Qin}
\email{webmaster@marysville-ohio.com}
\affiliation{%
  \institution{\\FORM, University of Southern Denmark}
  \city{Odense}
  \country{Denmark}
}

\renewcommand{\shortauthors}{Acclavio et al.}

\begin{abstract}
    The Endpoint Projection (EPP) theorem is a cornerstone of choreographic programming. 
    It states that every choreography can be projected to a network of processes that correctly implements it. 
    Proving EPP is notoriously difficult, and existing proofs are complex and non-modular because of the mismatch between the global view of choreographies and the local view of processes.
    
    In this article, we show how to reconcile this mismatch by designing a new semantics for choreographies that is built on the local view of processes, as well as a new preorder relation between choreographies and networks that extends bisimulation to deal with the propagation of knowledge of choice among distributed processes.
    As a result, we can give a modular proof of EPP, which is conceptually simpler than existing ones and also provides better insights on the theory of choreographic programming.
\end{abstract}


\received{20 February 2007}
\received[revised]{12 March 2009}
\received[accepted]{5 June 2009}

\maketitle

\section{Introduction}\label{sec:intro}
\newcommand{\client}{\pid{c}}
\newcommand{\service}{\pid{s}}
\newcommand{\cas}{\pid{cas}}

\paragraph{Background and motivation}
Choreographic Programming (CP) allows for programming concurrent and distributed systems from a global viewpoint \cite{montesi:phd,montesi:book}.
It works by providing a simple language for expressing the expected system behaviour that the programmer wants, and then shifting the burden of achieving a correct distributed implementation on a compilation procedure known as endpoint projection (EPP).
Proving EPP correct is typically done by establishing mathematical models of choreographies and networks of processes, and then demonstrating that EPP returns a network of process terms (or programs) that is semantically equivalent to the input choreography \cite{CM13,M13:phd,montesi:book,PGSN22,SHC25,PQM25,acc:man:mon:ESOP25}.
These proofs are notoriously difficult, because of a mismatch between the global view of choreographies and the local view of each process: the former has more information about causality.
The aim of this article is to develop an elegant solution to this dichotomy.

\paragraph{Challenge}
To understand the challenge, consider the next simple choreography ($\cC[fwd]$), where a process $\pid p$ communicates a number $n$ to a process $\pid q$, which then forwards this number to another process $\pid r$.
\[
\cC[fwd] = \chorpoly{\ccom pn qx}; \chorpoly{\ccom qx ry}
\]
The network returned by EPP for $\cC$, written $\proj{\cC}$, is the parallel composition of the three processes described in the choreography ($\pid p$, $\pid q$, and $\pid r$), each equipped with a program that executes the appropriate local actions: $\pid p$ sends $n$ to $\pid q$; $\pid q$ receives a value on its variable $x$ from $\pid p$, and then sends $x$ to $\pid r$; and $\pid r$ receives a value in its variable $y$ from $\pid q$.
\[
\proj{\cC[fwd]} = 
\pdo{\psend qn} \ppar \pdo[q]{\precv px; \psend rx} \ppar \pdo[r]{\precv qy}
\]
The definition of EPP is relatively simple and modular: the code for each process is compiled by recursively visiting the structure of the choreography, producing an action whenever the process is mentioned and skipping the choreographic instruction at the head otherwise.
The previous EPP is obtained by composing in parallel such `process projections' ($\proj[p]{\cC[fwd]}$, $\proj[q]{\cC[fwd]}$, and $\proj[r]{\cC[fwd]}$), as given next.
\begin{align*}
\proj[p]{\cC[fwd]} &= \psend qn &
\proj[q]{\cC[fwd]} &= \precv px; \psend rx &
\proj[r]{\cC[fwd]} &= \precv qy
\end{align*}
One would wish to prove the correctness of EPP compositionally, that is, by establishing first the correctness of each process projection.
Frustratingly, proofs in the literature do not proceed in this way.
The reason is that the semantics of a projected process does not necessarily correspond cleanly to the semantics of the choreography, since the choreography may contain more information.
For example, $\cC[fwd]$ can only perform the communication from $\pid p$ to $\pid q$, and only then can the second communication be executed.
Whereas $\proj[r]{\cC[fwd]}$ is immediately ready to perform the receive action to implement the second communication.
In other words, $\cC[fwd]$ sees that the receive at $\pid r$ depends on the completion of the first communication from $\pid p$ to $\pid q$, while $\proj[r]{\cC[fwd]}$ does not.
To solve this problem in the correctness proof of EPP, the standard solution is to reason about the entire projected network directly, falling short of a compositional approach.

Complexity increases when choreographies include choices: when a choreography chooses between the two branches of a choice, the processes whose behaviours depend on this choice must acquire knowledge of this choice through communication.
To understand this aspect, we borrow an example from \cite{montesi:book}, which we expand on later in \Cref{ex:chor}: in a single sign-on scenario, a client ($\pid c$) can gain access to a service ($\pid s$) by getting its credentials checked by a third-party central authentication service ($\cas$).
The $\cas$ decides whether $\pid s$ will communicate a session token to $\pid c$ or not.
Technically, the $\cas$ makes an internal choice, and then later communicates its decision to $\pid c$ and $\pid s$, so that they can know whether they should exchange a token or not.
At the choreographic level, it is easy to see that the `destiny' of $\pid c$ and $\pid s$ is decided as soon as $\cas$ makes its internal choice. But in the network produced by EPP, $\pid c$ and $\pid s$ remain available to both branches until they receive a communication.
This semantic discrepancy in how \emph{knowledge of choice} is propagated in choreographies (immediately) and networks (step by step) makes it challenging to relate choreographies to their EPP at runtime.
Previous work deals with this issue by developing syntactic relations that allow projected code to include `dead branches' at processes that will never actually be used \cite{CHY12,CM13,montesi:book}.
Thus, in proofs about EPP, even cases that should be simple are instead required to deal with entire equivalence classes (up to these relations).

\paragraph{This article}
In this article, we present a modular approach to the correctness proof of EPP.
Key to our method is the development of a new `local' semantics for choreographies. This semantics is parametrised on a process name, and allows for studying the local behaviour of a single process that participates in a choreography.
For example, under the local semantics for process $\pid r$, the previous choreography $\cC$ can immediately perform the receive action at $\pid r$ for the second communication.

Our local semantics has two important features.
\begin{enumerate}
    \item 
    The traditional global semantics of choreographies can be modularly reconstructed by aggregating the local semantics of all participating processes (\Cref{thm:agg-trad}).
    This result formally bridges reasoning about choreographies to the standard way of reasoning about networks.
    \item
    We can establish a new behavioural relation between the local semantics of a process $\pid p$ in a choreography $\cC$ and the semantics of a network $\nN$, written $\cC \mbbn[]{\pid p} \nN$, which guarantees that the implementation of $\pid p$ in the network complies with the choreography (\Cref{lem:mb_net}.\ref{lem:mb_net3}).
    This relation extends bisimulation to deal with the propagation of knowledge of choice among distributed processes, circumventing the need for introducing any syntactic relation for projections.
\end{enumerate}
We combine these features to devise a compositional proof of the correctness of EPP.
First, we establish that our relation supports modular reasoning: if a choreography is related to a network for all processes -- $\cC \mbbn[]{\pid p} \nN$ for all $\pid p$ -- then the choreography is (strongly) bisimilar to the network.
Second, we prove that process projection (the projection of a single process) is always related to its source choreography: it always holds that $\cC \mbbn[]{\pid p} \pdo{\proj[p]{\cC}}$ (\Cref{thm:cor_mbstrat_p_proj}).
Put together, these results imply the standard result that EPP guarantees bisimilarity between choreographies and the networks projected from them (\Cref{thm:epp}).

\subsection{Structure of the paper}
In \Cref{sec:chor}, we introduce the choreographic language together with its standard operational semantics, 
which we call \emph{traditional semantics}. We also introduce networks and their semantics, define endpoint projection, 
and the notion of bisimulation between choreographies and networks.
In \Cref{sec:newsem}, we define a new local semantics for choreographies. We show that the aggregate semamtics that naturally arises from the local semantics
coincides with the traditional semantics. 
In \Cref{sec:EPPnew}, we introduce a new behavioural relation between choreographies and networks, called \emph{branching preorder}, 
and establish its main properties. In particular, we prove that branching preorder implies bisimilarity (\Cref{thm:main}).
Finally, in \Cref{sec:epp_new}, we show that every choreography is in branching preorder with its projection. 
Combined with \Cref{thm:main}, this yields a new proof of the endpoint projection theorem.
\Cref{sec:related,sec:conclusion} discuss related work and conclusions.

\section{Background}\label{sec:chor}

In this section, we recall the syntax and semantics of choreographies we adopt in this paper,
we then recall the network calculus we use to model processes, and we conclude by recalling the traditional EPP theorem for choreographies.

\subsection{Choreographies: Syntax and Semantics}\label{subsec:syntax}
The set \defn{choreographies} $\ChorSet$ contains all terms generated from some pairwise disjoint sets of \defn{process names} ($\Pset$), \defn{variables} ($\Varset$), \defn{expressions} ($\Eset$), \defn{choreography variables} ($\CXset$), and \defn{\bllabels} ($\Lset=\set{\lleft,\lright}$, equipped with an involution $\overline{\cdot}$ such that $\overline{\lleft} = \lright$ and $\overline{\lright} = \lleft$)
using the following grammar:
\begin{equation}\label{eq:chor-syntax}
    \begin{array}{r@{\quad}r@{\;}c@{\;}l@{\quad}l}
        \text{\defn{Choreographies}}& 
            \cC,\cC_1,\cC_2 & \Coloneqq &
            \cnil \mid \cI \seq \cC \mid \cifte pe{\cC[1]}{\cC[2]} \mid \cX
            & 
        \\[5pt]
        \text{\defn{Instructions}}& 
            \cI &\Coloneqq & 
            \cgenloc \mid \cgencom \mid \cgensel
            & 
        \\[5pt]
        \text{where}&&& 
        \pid p, \pid q \in \Pset \text{ with } \pid p \neq \pid q
        \qomma
        x\in\Varset 
        \qomma
        \ee\in \Eset
        \qomma
        \cX \in \CXset 
        \qomma
        \albl\in\Lset 
    \end{array}
\end{equation}
That is, a choreography ($\cC$) can be:
\begin{itemize}
    \item a \defn{terminated choreography} $\cnil$;
    \item an \defn{instruction} ($\cI$) followed by a choreography $\cC$ (written $\cI \seq \cC$), where $\cI$ is the first instruction to be executed and $\cC$ is the continuation of the choreography after executing $\cI$. Each instruction can be one of the following:
    \begin{itemize}
        \item a \defn{local assignment} $\cgenloc$, where the process $\pid p$ evaluates the expression $\ee$ and stores the result in its local variable $x$;
        \item a \defn{communication} $\cgencom$, where the process $\pid p$ sends result of evaluating the expression $\ee$ to the process $\pid q$, which stores it in its local variable $x$; or
        \item a \defn{selection} $\cgensel$, where the process $\pid p$ communicates a selection label $\albl$ to the process $\pid q$.
    \end{itemize}
    \item a \defn{conditional statement} $\cifte pe{\cC[1]}{\cC[2]}$ where the process $\pid p$ decides the continuation $\cC[1]$ or $\cC[2]$ based on the evaluation of the expression $\ee$, or 
    \item a \defn{choreographic variable} ($\cX \in \CXset$).
\end{itemize} 

A \defn{(choreographic) procedure definition} is an equation $\genProcDef$ where $\cX \in \CXset$ and $\cC$ is a choreography, in which case we say that $\cX$ is \defn{defined} by $\genProcDef$.
A set of procedure definitions $\procC = \set{\procDef[1]{\cX}{\cC} \mid i \in I}$ is \defn{finitely closed} if it contains a finite number of procedure definitions and every choreography variable occurring in any $\cC_i$ is defined by a unique procedure definition in $\procC$.
\begin{remark}\label{rem:procChor}
    For simplicity, as in standard presentations of process calculi \cite{SW01,S11}, in the remainder we do not explicitly parametrise our choreographies on the set of procedure definitions. Instead, we assume that it is always finitely closed for the choreography under consideration and that it is guarded, as defined in \Cref{def:cset-guarded}.   
\end{remark}

\begin{example}\label{ex:chor}
    The following choreographies $\cC$ models a simple authentication protocol between a client $\pid c$ and a central authentication service $\cas$.
    The client  $\pid c$ sends its credentials $\creds$ to a central authentication service $\cas$ to gain access to a service $\pid s$.
    Then, $\cas$ can decide whether to validate the credentials or not, and accordingly inform $\pid c$ and $\pid s$ about the decision.
    If the credentials are valid, $\pid s$ will send a session token to $\pid c$, and $\cas$ will send a message $\okval$ to a logger $\logger$; 
    otherwise, $\cas$ will simply send a $\koval$ message to $\logger$, and no other communication will happen.
    \begin{equation}\label{eq:chor-ex}
        \begin{array}{rcl}
        \cC & = &
        \cccom c{\creds}{\cas}{x}
        \seq
        \cifte{\cas}{\valid(x)}{\cC[ok]}{\cC[ko]}
        \\
        \mbox{where}&
        \\
        \cC[ok] & = &
        \ccsel{\cas}{c}{l} \seq \ccsel{\cas}{s}{l} \seq \cccom{s}{\token()}{c}{t} \seq \cccom{\cas}{\okval}{\logger}{\resvar} \seq \cnil
        \\
        \cC[ko] & = &
        \ccsel{\cas}{c}{r} \seq \ccsel{\cas}{s}{r} \seq \cccom{\cas}{\koval}{\logger}{\resvar} \seq \cnil.
        \\
        \end{array}
    \end{equation}
    Notice that, while both $\pid c$ and $\pid s$ are informed about the decision of $\cas$ through the selection actions, $\logger$ is only involved in the protocol if one of the two branches is taken, and no knowledge of choice is needed at $\logger$ to correctly implement the choreography.
\end{example}

Given a choreography $\cC$, we write $\pnof{\cC}$ for the set of process names mentioned in $\cC$.
For all terms but choreography variables, $\pn$ is recursively defined as follows
\begin{align*}
    \pnof{\cnil} \defeq \emptyset
\qquad
    \pnof{\cI \seq \cC} = \pnof{\cI} \cup \pnof{\cC}
\qquad
    \pnof{\cifte pe{\cC[1]}{\cC[2]}} = \set{\pid p} \cup \pnof{\cC[1]} \cup \pnof{\cC[2]}
\end{align*}
where $\pnof{\cgenloc} = \set{\pid{p}}$ and $ \pnof{\cgencom} =\pnof{\cgensel} = \set{\pid{p}, \pid{q}}$.
For the case of choreography variables, we assume that for each $\cX$ in the set of choreographic procedures under consideration, the definition of $\pnof\cX$ is also given as an axiom.
In practice, for finitely closed sets of procedure definitions, $\pnof\cX$ can always be computed for all $\cX$ as shown in \cite{CMP23}.

\begin{definition}\label{def:cset-guarded}
    The set of \defn{active process names} $\apnof\cC$ of a choreography $\cC$ is defined as the set of process names that appear in $\cC$ before any procedure call.
    That is, 
    \begin{align*}
        \apnof{\cnil} = \apnof{\cX}=\emptyset 
        \qquad&\qquad
        \apnof{\cI \seq \cC} = \pn(\cI) \cup \apnof{\cC}
    \\
        \apnof{\cifte pe{\cC[1]}{\cC[2]}} &= \set{\pid p} \cup \apnof{\cC[1]} \cup \apnof{\cC[2]}
    \end{align*}
    We say that a set of procedure definitions is \defn{guarded} if $\apnof\cC \neq \emptyset$ for all procedure definitions $\genProcDef$ in the set.
\end{definition}

\begin{example}
To illustrates the difference between process names
and active process names, consider the
choreography
$ \cC = \cccom p m q x \seq \cX$ with procedure definition $\procDef{\cX}{\cC[1]}$
where $\cC[1] = \cccom p n r y \seq \cX$.
The choreography $\cC$ first lets $\pid p$ send a message to $\pid q$,
and then enters a recursive procedure in which $\pid p$ repeatedly
communicates with $\pid r$.

The set of process names of $\cC$ is $\pnof{\cC} = \set{\pid p,\pid q,\pid r}$ while the set of active process names is $\apnof{\cC} = \set{\pid p,\pid q}$.   
In fact, the only participants that are active at the top level in $\cC$ are those involved in the initial communication, while $\pid r$ occurs in the body of the recursive procedure.
\end{example}

We now define a \emph{symbolic} semantics for choreographies, that is, a semantics that does not evaluate local expressions at processes, following the presentation from \cite{acc:mon:per:OPDL}. This strengthens our development, because we will be able to show that EPP returns network that have exactly the same expressions at processes (local code is preserved). We later show how to modularly add expression evaluation to our semantics in \cref{sec:stores}, in a way that lifts all our results to that setting for free.

In \Cref{fig:trad-semantics} we present the (global) operational semantics for choreographies, 
where transitions are labeled by \defn{(global) labels} form the following set:
\begin{equation}\label{eq:Glab}
    \Glab =
    \setdef{
        \begin{array}{l@{\;,\;}l@{\;,\;}l}
            \TRcom pqx      & \TRsel pq &
        \\
            \TRassign[p] x{e} & \TRelse e &\TRthen e 
        \end{array}
    }{
        \;
        \pid p , \pid q \in \Pset, 
        \ee \in \Eset, 
        \albl \in \Lset,
        x \in \Varset
    }
\end{equation}
and where the function $\pnof\cdot$ is extended to labels as follows:
\[
    \pnof{\TRcom pqx} = \pnof{\TRsel pq} = \set{\pid{p} , \pid{q}}
\;\text{and}\;
    \pnof{\TRassign[p] xe} =\pnof{\TRthen e}= \pnof{\TRelse e} = \set{ \pid{p} } 
\]

\begin{remark}
    The semantics we define has two differences with previous presentations of global semantics for choreographies.
    The first is that we do not include in the semantics the runtime terms used to unfold procedure calls (see, e.g., \cite{montesi:book}), and instead directly unfold procedure calls in the semantics.
    The second is that, to keep the semantics in line with the one we use in the next sections, we do not take into account the state of the system, and instead only consider the \emph{symbolic semantics} of choreographies.
    This is a pedagogical choice, as in the next sections we will present our new semantics and proof strategy in the same setting to reduce the overhead of the presentation.
    However, in \Cref{sec:stores} we will extend back the semantics by taking into account the state of the system in a modular way over the symbolic semantics.
\end{remark}

\begin{figure}
    \def\myskip{\hskip8em}
    \centering
    $\begin{array}{c}
        \inferrule*[right=assign] {  }
        {\cgenloc \seq \cC  \tradsem{\TRassign[p]{x}{e}}  \cC }
    \myskip
        \inferrule*[right=call]{\genProcDef \\ \cC \tradsem{\mu} \cCp} 
        { \cX \tradsem{\mu}  \cCp }
    \\[15pt]
        \inferrule*[right=com] {  } 
        { \cgencom \seq \cC \tradsem{\TRcom pqx}  \cC }
    \myskip
        \inferrule*[right=sel] {  } 
        {\cgensel \seq \cC \tradsem{\TRsel pq} \cC}
    \\[15pt]
        \inferrule*[right=cond-then] { } 
        {\cgenifte \tradsem{\TRthen e}  \cC[1]  }
    \qquad
        \inferrule*[right=cond-else] { } 
        { \cgenifte\tradsem{\TRelse e}  \cC[2]  }
    \\[15pt]
        \inferrule*[right=delay] { \cC \tradsem{\mu} \cCp \\ \pn(\cI)\disj\pn(\mu)} 
        { \cI\seq \cC \tradsem{\mu}  \cI\seq \cCp}
    \qquad
        \inferrule*[right=delay-cond] { \cC[1] \tradsem{\mu}  \cCp[1] \\  \cC[2] \tradsem{\mu}  \cCp[2] \\ \pid{p}\notin \pn(\mu)}
        { \cgenifte \tradsem{\mu}  \cgeniftep } 
    \\[15pt]
    \end{array}$

    \caption{Traditional operational semantics for choreographies.}
    \label{fig:trad-semantics}
\end{figure}

\subsection{Processes and Networks: Syntax and Semantics}\label{subsec:Networks}
We define \defn{process terms} from the same sets of process names ($\Pset$), variables ($\Varset$), expressions ($\Eset$), and \bllabels ($\Lset=\set{\lleft,\lright}$) as choreographies, together with a set of \defn{process variables} ($\PXset$) using the following grammar.
\begin{equation}\label{eq:proc-syntax}
    \begin{array}{r@{\quad}r@{\;}c@{\;}l}
        \text{\defn{Process terms}}& \pP,\pP[1],\pP[2] & \Coloneqq &
            \pnil               \mid 
            \palpha \pseq \pP    \mid
            \genplabel          \mid
            \genpifte           \mid 
            \pX 
        \\[5pt]
        \text{\defn{Actions}}& \palpha &\Coloneqq & 
            \;\pgenass            \;\mid 
            \;\psend{q}{e}        \;\mid 
            \;\precv{p}{x}        \;\mid 
            \;\psendl[\albl]{q}
        \\[5pt]
        \text{where}&&& 
        \pid p, \pid q \in \Pset \text{ with } \pid p \neq \pid q
        \qomma
        x\in\Varset 
        \qomma
        \ee\in \Eset
        \qomma
        \albl\in\Lset 
        \text{\;, and\;}
        \emptyset \neq \ELLset \subseteq \Lset
    \end{array}
\end{equation}

A \defn{(process) procedure definition} is an equation $\genpProcDef$ where $\pX \in \PXset$ and $\pP$ a process term, in which case we say that $\pX$ is \defn{defined} by $\genpProcDef$.
As for sets of choreographic procedures definitions (see \Cref{rem:procChor}), in the following we assume that every time we give a process term $\pP$, we also assume a \defn{finitely closed} set $\procP$ of procedure definitions to be given, such that every process variable occurring in $\pP$ occurs in a procedure definition in $\procP$, as well as every process variable occurring in the definition of $\pnof{\pX}$ for all $\pX$ in $\procP$.

We extend the function $\pn$ to actions and process terms as follows:
\begin{equation}
    \begin{array}{c}
        \pnof{\pnil}   = \emptyset 
    \qquad
        \pnof{\palpha \pseq \pP} = \pnof{\palpha} \cup \pnof{\pP}
    \qquad
        \pnof{\genplabel} = \set{\pid p} \cup \bigcup_{\albl \in \Lset}\pnof{\pP}
    \\
        \pnof{\genpifte} = \pnof{\pP[1]} \cup \pnof{\pP[2]}
    \qquad
        \pnof{\pX}     = \pnof{\pP} \quad\text{if } \pX \metaDef P \in \procP
    \end{array}
\end{equation}

The set of \defn{networks} $\NetSet$ contains all functions from process names to process terms with finite support (only finitely many processes are mapped to terms that are not $\pnil$), such that no process name $\pid p$ is mapped to a process term mentioning $\pid p$ itself.
A network is \defn{atomic} if the set of process names is a singleton.
We denote by $\nnil$ any \defn{terminated network} mapping all process names to the process $\pnil$.
If $\nN[1]$ and $\nN[2]$ are network with disjoint supports (i.e., $\supp{\nN[1]} \cap \supp{\nN[2]} = \emptyset$), their \defn{parallel composition}, denoted $\nN[1] \ppar \nN[2]$, is the network defined as the union of the two functions. Parallel composition is a commutative monoid for disjoint networks \cite{montesi:book}.
The definition of $\pn$ is further extended to networks by letting $\pnof\nN=\bigcup_{\pid p \in \dom{\nN}} \pnof{\nN(\pid  p)}$.
In the reminder, we assume that every time we give a network $\nN$, we also assume a finitely closed set $\procP$ of procedure definitions to be given, such that every process variable occurring in $\nN$ occurs in a procedure definition in $\procP$.

\begin{example}\label{ex:network}

    The network $\nN$ obtained projecting the choreography of \Cref{ex:chor} is defined as follows:

$$
\nN(\pid p) =
\begin{cases}
{
\psend{cas}{\creds}
\pseq
\plabels{cas}{
\lbl{l} :
\precv{s}{t}
\pseq
\pnil,
\lbl{r} :
\pnil
}
}
& \text{if } \pid p = \pid c,
\\
{
\precv{c}{x}
\pseq
\pifte{\valid(x)}
{
Q_{\lleft}
}
{
Q_{\lright}
}
}
& \text{if } \pid p = \cas,
\\
{
\plabels{cas}{
\lbl{l} :
\psend{c}{\token()}
\pseq
\pnil,
\lbl{r} :
\pnil
}
}
& \text{if } \pid p = \pid s,
\\
{
\precv{cas}{\resvar}
\pseq
\pnil
}
& \text{if } \pid p = \logger,
\\
\pnil
& \text{otherwise.}
\end{cases}
$$

where $\pQ[{\lleft}] = \psendl[l]{c} \pseq \psendl[l]{s} \pseq \psend{logger}{\okval} \pseq \pnil$ and  $\pQ[{\lright}] = \psendl[r]{c} \pseq \psendl[r]{s} \pseq \psend{logger}{\koval} \pseq \pnil$.

\end{example}

In \Cref{fig:network:semantics} we recall the semantics for networks.
We distinguish between two kinds of transitions: \defn{local} transitions involving a single process (the one highlighted in the label). and \defn{global} transitions, where more than one process are involved.
The local transitions have the following interpretation:
\begin{itemize}
    \item The rule \textsc{assign} states that a transition with label $\PRassign[p]{x}{e}$ is performed at the process $\pid{p}$, 
    executing an assignment action $\pgenass$ and continuing with the process term $P$.
    Note that the highlighted process name $\pid{p}$ in the transition label denotes the process on which the action is being performed.
     
    \item 
    The rule \textsc{sendVal} states that a transition with label $\PRsend pq$ is performed at the process $\pid{p}$, 
    executing a send action $\psend{q}{e}$ 
    (where $\pid{p}$ is the sender, $e$ is the expression being sent, and $\pid{q}$ is the receiver),
    and continuing with the process term $\pP$.
    
    \item 
    Dually, the rule \textsc{recvVal} states a transition with label $\PRrecv pq$ is performed at the process $\pid{q}$, 
    executing a receive action $\precv{p}{x}$
    (where $\pid{q}$ is the receiver, $x$ is the variable where the received value is stored, and $\pid{p}$ is the sender), 
    and continuing with the process term $\pQ$.

    \item 
    Similarly, the rules \textsc{sendLbl} and \textsc{recvLbl} model the sending and receiving of labels for selection and branching.

    \item 
    Rules \textsc{cond-then} and \textsc{cond-else} model the execution of a conditional statement, where the process $\pid{p}$ evaluates the expression $e$ and continues with $\pP_1$ if $e$ is true, and with $\pP_2$ if $e$ is false.
    
\end{itemize}
The global transitions of parallel composed networks have the following interpretation:
\begin{itemize}
    \item The rule \textsc{com} states if a network $\nN[1]$ can perform a transition on the process $\pid{p}$ for sending an expression to $\pid{q}$
    and a network $\nN[2]$ can perform a transition on the process $\pid{q}$ for receiving an expression and storing it in the variable $x$ from $\pid{p}$, 
    then the parallel composition of the two networks can perform a transition for the communication of a value from $\pid{p}$ to $\pid{q}$.

    \item 
    Similarly, the rule \textsc{sel} states that if a network $\nN[1]$ can perform a transition on the process $\pid{p}$ for sending a label $\albl$ to $\pid{q}$ and a network $\nN[2]$ can perform a transition on the process $\pid{q}$ for receiving the same label $\albl$ from $\pid{p}$, then the parallel composition of the two networks can perform a transition for the selection of label $\albl$ from $\pid{p}$ to $\pid{q}$.
    
    \item The rule \textsc{par-left} models a parallel execution. It states that if a network $\nN[1]$ can perform a transition, 
    then the larger network $\nN[1] \ppar \nN[2]$ can perform the same transition. The part of network $\nN[2]$ not involved in the transition remains unchanged.
\end{itemize}

\begin{figure}[t]
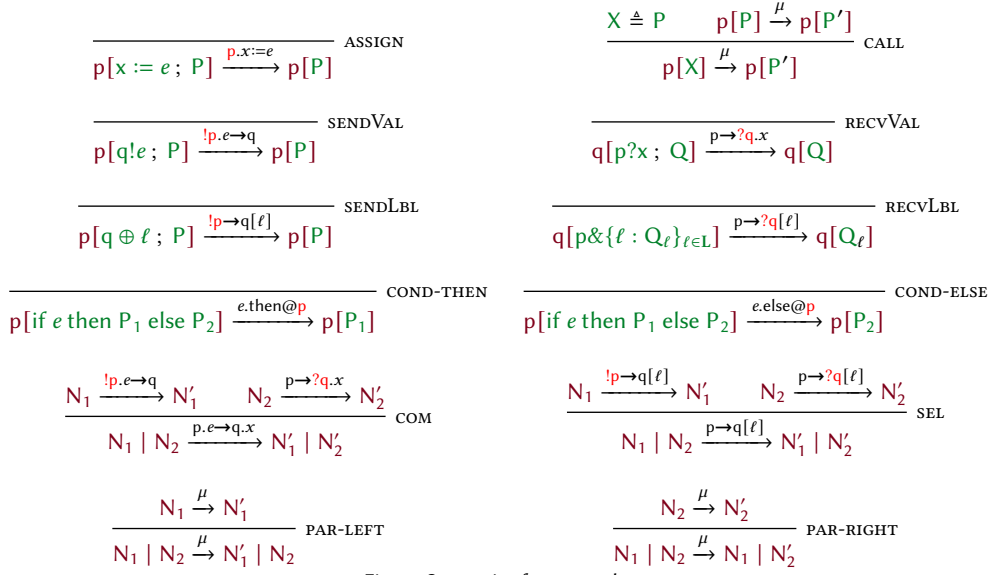

    \adjustbox{max width=\textwidth}{$\begin{array}{cc}
        \inferrule*[right=assign]
        { }
        { \pdo{\pgenass \pseq \pP} \proctr{\PRassign[p]{x}{e}} \pdo{\pP} }
    &
        \inferrule*[right=call]
        { \genpProcDef \\  \pdo{\pP} \proctr{\mu} \pdo{\pP'} }
        { \pdo{ \pX }
        \proctr{\mu} \pdo{\pP'} }
    \\\\
        \inferrule*[right=sendVal]
        { }
        { \pdo{\psend{q}{e} \pseq \pP} \proctr{\PRsend pq} \pdo{\pP} }
    &
        \inferrule*[right=recvVal]
        { }
        { \pdo[q]{\precv{p}{x} \pseq \pQ} \proctr{\PRrecv pq} \pdo[q]{\pQ} }
    \\\\
        \inferrule*[right=sendLbl]
        { }
        { 
            \pdo{\psendl[\albl]{q} \pseq \pP} 
            \proctr{\PRlsend pq} 
            \pdo{\pP} 
        }
    &
%
        \inferrule*[right=recvLbl]
        { }
        {  
            \pdo[q]{\plabel[\pQ] p\albl{\Lset}} 
            \proctr{\PRlrecv pq} 
            \pdo[q]{\pQ_{\albl}} 
        }
    \\\\
        \inferrule*[right=cond-then]
        { }
        { \pdo{\genpifte}
        \proctr{\PRthen[p]{e}} \pdo{\pP[1]} }
    &
        \inferrule*[right=cond-else]
        { }
        { \pdo{\genpifte}
        \proctr{\PRelse[p]{e}} \pdo{\pP[2]} }
    \\\\
        \inferrule*[right=com]
        { \nN[1] \netsem{\PRsend pq} \nNp[1] \\
        \nN[2] \netsem{\PRrecv pq} \nNp[2] }
        { \nN[1] \ppar \nN[2] \netsem{\TRcom pqx} \nNp[1] \ppar \nNp[2] }
    &
        \inferrule*[right=sel]
        { \nN[1] \netsem{\PRlsend[\albl]{p}{q}} \nNp[1] \\
        \nN[2] \netsem{\PRlrecv[\albl]{p}{q}} \nNp[2] }
        { \nN[1] \ppar \nN[2] \netsem{\TRsel pq} \nNp[1] \ppar \nNp[2] }
    \\\\
        \inferrule*[right=par-left]
        { \nN[1] \netsem{\mu} \nNp[1] }
        { \nN[1] \ppar \nN[2] \netsem{\mu} \nNp[1] \ppar \nN[2] }
    &
        \inferrule*[right=par-right]
        { \nN[2] \netsem{\mu} \nNp[2] }
        { \nN[1] \ppar \nN[2] \netsem{\mu} \nN[1] \ppar \nNp[2] }
    \end{array}$}

    \caption{Semantics for networks.}
    \label{fig:network:semantics}
\end{figure}

We then have the result, which ensures that the global semantics of networks is consistent with the local semantics of processes.
That is, if a network $\nN$ performs a transition $\mu$ involving only one process $\pid p$, then the same transition with label $\mu$ can be performed on the atomic network $\pdo{\nNof{\pid p}}$
and vice versa.

\begin{lemma}\label{lem:net_sem}
    Let $\nN$ and $\nNp$ be networks, and $\mu$ such that \emph{$\pnof{\mu} = \set{\pid p}$}.
    Then, 
    $\nN \netsem{\mu} \nNp$ if and only if $\pdo{\nNof{\pid p}} \netsem{\mu} \pdo{\nNof{\pid p}}$.
\end{lemma}
\begin{proof}
    By definition of the semantics of processes and networks, both directions of the statement hold by induction on the structure of the network $\nN$.
\end{proof}

As a consequence, if a network $\nN$ has a process $\pid p$ that is not terminated, then there is a transition of $\nN$ involving $\pid p$.
\begin{lemma}\label{lem:local_progress_net}
   Let $\nN$ be a network.
   If $\nN(\pid p) \neq \pnil$,
   then $\nN \netsem{\mu} \nNp$ for some $\nNp$ and $\mu$ such that \emph{$\pnof{\mu} = \set{\pid{p}}$}.
\end{lemma}
\begin{proof}
    By \Cref{lem:net_sem}, definition of network semantics (see \Cref{fig:network:semantics}) and guardedness. 
\end{proof}

\subsection{Endpoint Projection}\label{subsec:EPP}

We conclude this section by recalling the definition of \defn{endpoint projection} (EPP) for choreographies.
For this purpose, we first recall the definition of the local projection of a choreography on a process name, the definition of the merge operator on process terms allowing us to define the projection of a choreography, and finally the endpoint-projection (EPP) as the parallel composition of the local projections of a choreography on all process names involved in the choreography, where the merge operation is used to handle the knowledge of choice when defining local projection.
We then recall the EPP theorem, stating that the endpoint projection of a choreography is bisimilar to the original choreography, and we sketch the idea of the proof.

\begin{definition}\label{def:projection}
    Let $\cC$ be a choreography.
    Its \defn{endpoint projection} $\proj{\cC}$ is the network 
    \begin{equation}
        \proj\cC = \pdo[\pid p_1]{\proj[\pid p_1]{\cC}} \ppar \cdots \ppar \pdo[\pid p_n]{\proj[\pid p_n]{\cC}}
        \quad\text{where } \pnof{\cC} = \set{\pid p_1, \ldots, \pid p_n}
    \end{equation}
    associating to each process name $\pid p\in\pnof{\cC}$ the process term $\proj[p]{\cC}$ defined as in \Cref{fig:projection} using the axiliary \defn{merge operator} $\sqcup$ defined in the same figure.
    We say that $\cC$ is \defn{projectable at} a process $\pid p$ if $\proj[p]{\cC}$ is defined, and that $\cC$ is \defn{projectable} if it is projectable at all process names $\pid p \in \pnof{\cC}$.
    We may say that an action $\palpha$ is compatible with a choreographic instruction $\cI$ at a process $\pid p$ whenever $\proj[p]{\cI \pseq \cnil}=\palpha \pseq \pnil$.

    The \defn{projection} $\proj{\procC}$ of a set of choreographic procedure definitions $\procC$ on a process name $\pid{p}$ is defined as the set of process procedure definitions associating to each choreographic procedure definition $\cX_i$ occurring in $\procC$ and each process name $\pid p$ occurring in $\cC$, a process procedure definition defining a process variable $\pX^{\pid p}_i$ as the projection of $\cC_i$ on $\pid p$.
    That is, $\proj[\pid p]{\procC} = \set{\procDef{{\pX^{\pid p}_i}}{{\proj[p]{{\cC_i}}}} \mid \genProcDef[i] \in \procC \mbox{ and } \pid{p} \in \pnof{\cC}}$.
\end{definition}

\begin{figure}[t]
    \adjustbox{max width = \textwidth}{$\begin{array}{c}
        \proj[r]{\cnil} 
    =
        \pnil 
    \qquad
        \proj[r]{\cifte{p}{e}{\cC[1]}{\cC[2]}}
    =
        \begin{cases}
            \pifte e{\proj[r]{\cC[1]}}{\proj[r]{\cC[2]}}
            & \text{if } \pid{r} = \pid{p} \\
            \proj[r]{\cC[1]} \sqcup \proj[r]{\cC[2]}
            & \text{otherwise}
        \end{cases} 
    \\\\
        \proj[r]{\cgenloc \seq \cC} 
    =
        \begin{cases}
            \pgenass \pseq \proj[r]{\cC} & \text{if } \pid{r} = \pid{p} \\
            \proj[r]{\cC}                        & \text{otherwise}
        \end{cases} 
    \qquad
        \proj[r]{\cX}
    =
        \begin{cases}
            \pX^{\pid{r}} & \text{if } \pid{r} \in \pnof{\cX} \\
            \pnil & \text{otherwise}
        \end{cases}
    \\\\
        \proj[r]{\cgencom \seq \cC}
    =
        \begin{cases}
            \psend{q}{e} \pseq \proj[r]{\cC}  & \text{if } \pid{r} = \pid{p} \\
            \precv{p}{x} \pseq \proj[r]{\cC}  & \text{if } \pid{r} = \pid{q} \\
            \proj[r]{\cC}                        & \text{otherwise}
        \end{cases} 
    \qquad
        \proj[r]{\cgensel \seq \cC}
    =
        \begin{cases}
            \psendl[\albl]{q} \pseq \proj[r]{\cC} 
            & \text{if } \pid{r} = \pid{p} \\
            \ponelabel[{\proj[r]{\cC}}] p\albl
            & \text{if } \pid{r} = \pid{q} \\
            \proj[r]{\cC}
            & \text{otherwise}
        \end{cases} 
    \\\\\hline\\
        \pP \sqcup \pQ 
    =
        \begin{cases}
            \pnil  
            & \text{if } \pP = \pnil \text{ and } \pQ = \pnil
        \\
            \pX  
            & \text{if } \pP = \pX \text{ and } \pQ = \pX 
        \\
            \palpha \pseq \procpoly{(\pPp \sqcup \pQp)  }
            & \text{if } \pP = \palpha \pseq \pPp \quand \pQ = \palpha \pseq \pQp 
        \\
            \ponelabel[\pPp \sqcup \pQp] p\albl 
            & \text{if } \pP =  \ponelabel[\pPp] p\albl,\quand
                \pQ = \ponelabel[\pQp] p\albl 
        \\
            \precl[]{p} \procpoly{\set{\lleft : \pPp, \lright : \pQp}}
            & \text{if } \pP =  \ponelabel[\pPp] p\lleft
            \quand
                \pQ = \ponelabel[\pQp] p\lright 
        \\
            \precl[]{p} \procpoly{\set{\lleft : \pQp, \lright : \pPp}}
            & \text{if } \pP =  \ponelabel[\pPp] p\lright 
            \quand
                \pQ = \ponelabel[\pQp] p\lleft
        \\
            \precl[]{p} \procpoly{\set{\lleft : \pPp \sqcup \pQ[{\lleft}], \lright : \pQ[\lright]}}
            & \text{if } \pP =  \ponelabel[\pPp] p\lleft
            \quand
                \pQ = \precl[]{p} \procpoly{\set{\lleft : \pQ[\lleft], \lright : \pQ[\lright]}}  
        \\
            \precl[]{p} \procpoly{\set{\lleft :  \pQ[{\lleft}], \lright : \pQ[\lright] \sqcup \pPp}}
            & \text{if } \pP =  \ponelabel[\pPp] p\lright
            \quand
                \pQ =\precl[]{p} \procpoly{\set{\lleft : \pQ[\lleft], \lright : \pQ[\lright]}}
        \\
              \precl[]{p} \procpoly{\set{\lleft :  \pP[{\lleft}] \sqcup \pQp , \lright : \pP[\lright]}}
            & \text{if } \pP =  \precl[]{p} \procpoly{\set{\lleft : \pP[\lleft], \lright : \pP[\lright]}}
            \quand
                \pQ = \ponelabel[\pQp] p\lleft  
        \\
            \precl[]{p} \procpoly{\set{\lleft :  \pP[{\lleft}], \lright : \pP[\lright] \sqcup \pQp}}
            & \text{if } \pP =  \precl[]{p} \procpoly{\set{\lleft : \pP[\lleft], \lright : \pP[\lright]}}
            \quand
                \pQ =\ponelabel[\pQp] p\lright
        \\
            \palabel[\pPp \sqcup \pQp] p\albl\Lset
            & \text{if } \pP =  \palabel[\pPp] p\albl\Lset 
            \quand
                \pQ = \palabel[\pQp] p\albl\Lset 
        \\
            \pifte e {\procpoly{(\pP[1] \sqcup \pQ[1])}}{\procpoly{(\pP[2] \sqcup \pQ[2])}}
            & 
                \pP=\pifte e{\pP[1]}{\pP[2]}
            \quand
                \pQ = \pifte e {\pQ[1]}{\pQ[2]} 
        \\
            \text{undefined}
            & \text{otherwise}
        \end{cases}
    \end{array}$}
    \caption{Definition of process projection of a choreography $\cC$ on a process name $\pid{p}$, and the definition of the merge operation $\sqcup$ on process terms required for the case of the conditional.
    }
    \label{fig:projection}
\end{figure}

\begin{remark}
    The need for the merge operation arises from the fact that, when projecting a choreography on a process name $\pid{p}$, we need to handle the case of conditionals in the choreography.
    In particular, when a choreography has a conditional contructor \emph{$\cifte pe{\cC[1]}{\cC[2]}$} and a process name $\pid{q}$ is not involved in the evaluation of the condition $\ee$, then $\pid{q}$ cannot know which branch of the conditional will be executed.
    The merge is a partial function that takes two process terms and merges them into a single process term.
    It is crucial for handling the knowledge of choice when defining process projection.
    In general, merging two terminating processes $\pnil$ results in $\pnil$ and merging two process variables $\pX$ results in $\pX$.
    Merging two process terms of the form $\alpha \pseq P$ where the action $\alpha$ is not a receive label requires that the two process terms have the same action $\alpha$.
    The resulting process term starts with the same action $\alpha$ and continues with the merge of the continuations of the two process terms.
    Merging two process terms of receiving labels has several cases. The basic idea is that, when the two process terms have the same label,
    the two process terms under the same labels are merged. When the two process terms have different labels,
    the resulting process terms has both labels and their associated process terms.
    Merging two conditional process terms requires that the two process terms have the same condition
    and the merge of the two process terms in the then branch and the else branch are defined.
    The resulting process term has the same condition and the merge of the two process terms in the then branch and the else branch as its branches.

\end{remark}

\begin{example}\label{ex:EPP}
    The EPP of the choreography from \Cref{ex:chor} is the network in \Cref{ex:network}.
\end{example}

We recall the following properties of the merge operation, which we will use later in some proofs.
\begin{proposition}\label{prop:merge}
    The merge operator
    is a partial join operator on process terms respecting the following properties:
    \begin{align*}
    \pP \sqcup \pP &= \pP & \text{(Idempotency)} \\
    \pP \sqcup \pQ &= \pQ \sqcup \pP & \text{(Commutativity)} \\
    \pP \sqcup (\pQ \sqcup \pR) &= (\pP \sqcup \pQ) \sqcup \pR & \text{(Associativity)}
    \end{align*}
\end{proposition}
\begin{proof}
    By structural induction on the process terms. 
\end{proof}

By definition of the merge operator, we have the following properties relating the merge operator and the local semantics of processes, which we will use later in some proofs.
\begin{lemma}\label{lem:merge_and_local}
    Let $\pP[1]$ and $\pP[2]$ be process terms such that $\pP[1] \sqcup \pP[2]$ is defined
    and $\pid p \in\Pset $ such that \emph{$\pid p \in \pnof{\pP[i]}$} for $i \in \set{1,2}$.
    Then,
    \begin{enumerate}
        \item\label{lem:merge_local_sem1} 
        if $\pdo[p]{\pP[1]} \netsem{\mu} \pdo[p]{\pPp[1]}$ and $\pdo[p]{\pPp[2]} \netsem{\mu} \pdo[p]{\pPp[2]}$, 
        then $\pdo[p]{\pP[1] \sqcup \pP[2]} \netsem{\mu} \pdo[p]{\pPp[1] \sqcup \pPp[2]}$;
        
        \item\label{lem:merge_local_sem2} 
        if $\pdo[p]{\pP[1]} \netsem{\mu} \pdo[p]{\pPp[1]}$ and $\pdo[p]{\pP[2]} \notnetsem{\mu}$, 
        then $\pdo[p]{\pP[1] \sqcup \pP[2]} \netsem{\mu} \pdo[p]{\pPp[1]}$.

        \item\label{lem:merge_local_sem5} 
        if $\pdo[p]{\pP[1] \sqcup \pP[2]} \netsem{\mu} \nNp$, then either:
        \begin{itemize}
            \item 
            both 
            $\pdo[p]{\pP[1]} \netsem{\mu} \pdo[p]{\pPp[1]}$ and 
            $\pdo[p]{\pP[2]} \netsem{\mu} \pdo[p]{\pPp[2]}$, 
            and $\nNp = \pdo[p]{\pP[1] \sqcup \pP[2]}$;
            or

            \item 
            $\mu = \PRlrecv[$\albl$]{q}{p}$, 
            and we have $\pdo[p]{\pP[i]} \netsem{\mu} \pdo[p]{\pPp[i]}$ and
            $\pdo[p]{\pP[j]} \notnetsem{\mu}$, 
            but 
            $\pdo[p]{\pP[j]} \netsem{\PRlrecv[$\overline{\albl}$]{q}{p}} \pdo[p]{\pPp[j]}$ 
            with $\set{i,j} = \set{1,2}$.
        \end{itemize} 
    \end{enumerate}
\end{lemma}
\begin{proof}
    By case analysis on the definition of $\sqcup$.
\end{proof}

\subsection{Bisimulation and EPP Theorem}\label{subsec:EPPthm}
We conclude this section by recalling the definition of bisimulation (between choreographies and networks), and the statement of the EPP theorem, which states that the endpoint projection of a choreography is bisimilar to the original choreography.

\begin{definition}\label{def:bisim}
    A relation $\relSymb \subset( \ChorSet \times \NetSet )$ between the set of choreographies and the set of networks is 
    a \defn{bisimulation} if the following hold for any choreograpy $\cC$, any network $\nN$, and any label $\mu\in\Glab$:
    \begin{enumerate}
        \item\label{def:bisim:1} 
        if $\cC \tradsem{\mu} \cCp$ and $\cC\relSymb \nN$, 
        then $\nN \netsem{\mu} \nNp $ and $\cCp \ \mathcal{R} \ \nNp$;

        \item\label{def:bisim:2} 
        if $\nN \netsem{\mu} \nNp $ and $\nN\relSymb \cC$, 
        then $\cC \tradsem{\mu} \cCp$ and $\cCp \ \mathcal{R} \ \nNp$.
    \end{enumerate}
    A choreography $\cC$ and a network $\nN$ are \defn{bisimilar}, written $\cC \bisim \nN$ if there is 
    a bisimulation $\relSymb$ such that $\cC \relSymb \nN$.
\end{definition}

\begin{restatable}[EPP Theorem]{theorem}{EPP}\label{thm:EPPtrad}
    Let $\cC$ be choreography. If $\cC$ is projectable, then $\cC \bisim \proj[]{\cC}$.
\end{restatable}
The textbook proof of the EPP theorem \cite{montesi:book} -- and proofs in the literature based on the same ideas -- incurs the problems discussed in the introduction: it is not modular and requires dealing with the equivalence classes given by the merge operator.

We conclude this section by recalling the standard construction of the stratification of bisimulation and the textbook result (see \cite[Theorem~2.10.13]{S11}) that states that the stratification of bisimulation coincides with bisimulation itself, and therefore can be used as an alternative characterisation of bisimulation.
This result will be used in the proof of our main Theorem of \Cref{sec:MBandEPP}.
\begin{definition}\label{def:bisim_strat}
    Let $n \in \mathbb N$.
    We define the \defn{$n$-th stratification of bisimulation} relation $\bisimn\subseteq \ChorSet \times \NetSet$ by induction on $n$ as follows:
    \begin{itemize}
        \item $\cC \bisimn[0] \nN$ for all $\cC$ and $\nN$;
        \item $\cC \bisimn[n+1] \nN$ the following hold for all $\mu\in\Glab$:
        \begin{itemize}
            \item $\cC \tradsem{\mu} \cCp$ implies $\nN \netsem{\mu} \nNp$ and $\cCp \bisimn[n] \nNp$; and
            \item $\nN \netsem{\mu} \nNp$ implies $\cC \tradsem{\mu} \cCp$ and $\cCp \bisimn[n] \nNp$.
        \end{itemize}
    \end{itemize}

    The \defn{stratification of bisimulation} relation $\bisimn[\omega]$ as the limit of the stratification of bisimulation relations $\bisimn[n]$ as $n$ goes to infinity.
    That is, $\cC \bisimn[\omega] \nN$ if and only if $\cC \bisimn[n] \nN$ for all $n \in \mathbb N$.
\end{definition}
\begin{theorem}\label{thm:bisim_strat}
    The stratification of bisimulation ($\bisimn[\omega]$) coincides with bisimulation ($\bisim$).
\end{theorem}

\section{From Local to Global: a new semantics for choreographies}\label{sec:newsem}

In this section, we define a new semantics for choreographies,
The main motivation for defining this new semantics is to have a semantics for choreographies that is more modular and compositional than the traditional semantics, and that is more convenient for proving the EPP theorem.
The idea is to define the semantics of choreographies by aggregating the transitions of a new process-oriented semantics, where each step is labelled by a local label, that is, a label that describes the action performed by a single process.
This approach is intended to replicate the way the semantics of processes is typically defined in process calculi \cite{M80,SW01,S11}, following the well-known approach of Structural Operational Semantics (SOS) \cite{P04a}, in order to simplify relating the semantics of choreographies to that of process terms.

\begin{figure}[t]
    \def\myskip{\hskip6em}
    \adjustbox{max width=\textwidth}{$
    \begin{array}{c}
        \inferrule*[right=$\pid r$ assign] {  }
        {\cloc rxe \seq \cC \localsem{\TRassign x e}  \cC}
    \myskip
        \inferrule*[right=$\pid r$ call]{\genProcDef \\ \cC \localsem{\mu} \cCp} 
        { \cX \localsem{\mu}  \cCp }
    \\\\
        \inferrule*[right=$\pid r$ send] {  } 
        {\cccom re qx \seq \cC  \localsem{\PRsend rq}  \cC }
    \myskip
        \inferrule*[right=$\pid r$ rec] {  } 
        {  \cccom pe rx \seq \cC       \localsem{\PRrecv pr}   \cC     }
    \\\\
        \inferrule*[right=$\pid r$ sendl] {  } 
        {  \ccsel rq \albl \seq \cC     \localsem{\PRlsend rq}   \cC     }
    \myskip
        \inferrule*[right=$\pid r$ recl] {  } 
        {  \ccsel pr \albl \seq \cC      \localsem{\PRlrecv pr}   \cC     }
    \\\\  
        \inferrule*[right=$\pid r$ cond-then] { } 
        {  \cifte re{\cC[1]}{\cC[2]}      \localsem{\TRthen[r]{e}}   \cC[1]      }
    \qquad
        \inferrule*[right=$\pid r$ cond-else] { } 
        {  \cifte re{\cC[1]}{\cC[2]}      \localsem{\TRelse[r]{e}}   \cC[2]      }
    \\\\
        \inferrule*[right=$\pid r$ delay] {  \cC      \localsem{\mu}   \cCp     \\ \pid r \notin \pn(\cI)} 
        {  \cI\seq \cC      \localsem{\mu}   \cI\seq \cCp   }
    \qquad
        \inferrule*[right=$\pid r$ delay-cond] {  \cC[1]      \localsem{\mu}   \cCp[1]      \\   \cC[2]      \localsem{\mu}   \cCp[2]     \\ \pid{p}\neq \pid{r}}
        {  \cifte pe{\cC[1]}{\cC[2]}      \localsem{\mu}   \cifte pe{\cCp[1]}{\cCp[2]}    }
    \\\\
        \inferrule*[right=$\pid r$ delay-predict-left] {  \cC[1]      \localsem{\PRlrecv qr}   \cCp[1]      \\   \cC[2] \localsem[r]{\PRlrecv[$\overline{\albl}$] qr} \cCp[2]  \\ \pid{r} \neq \pid{p} }
        {  \cifte pe{\cC[1]}{\cC[2]}      \localsem{\PRlrecv qr}   \cCp[1]}
    \\\\
        \inferrule*[right=$\pid r$ delay-predict-right] {  \cC[1]      \localsem{\PRlrecv[$\overline{\albl}$] qr}   \cCp[1]      \\   \cC[2] \localsem[r]{\PRlrecv qr} \cCp[2]  \\ \pid{r} \neq \pid{p}}
        {  \cifte pe{\cC[1]}{\cC[2]}      \localsem{\PRlrecv qr}   \cCp[2]}
    \\\\\hline\\
        \inferrule*[right=local] {  \cC      \localsem{\mu}   \cCp     \\ \mu \in \set{\TRassign xe , \ \TRthen[r]{e} , \TRelse[r]{e} \mid \ee \in \Eset, x \in \Varset}}
        {  \cC       \aggsem{\mu}   \cCp   }
    \\\\
        \inferrule*[right=agg-com] {    \cC      \localsem[p]{\PRsend pq}    \cCp    \\   \cC     \localsem[q]{\PRrecv pq}   \cCp     }
        {    \cC       \aggsem{\TRcom pqx}    \cCp    }
    \qquad
        \inferrule*[right=agg-sel] {  \cC      \localsem[p]{\PRlsend pq}   \cCp      \\ 
        \cC     \localsem[q]{\PRlrecv pq}   \cCp     }
        {   \cC      \aggsem{\TRsel pq}   \cCp     }
    \\\\
    \end{array}
    $}

    \caption{
        Rules for the local and aggregate operational semantics for choreographies.
    }
    \label{fig:new-semantics}
\end{figure}

\subsection{Local and aggregated symbolic semantics for choreographies}\label{subsec:local}

We define new \emph{local} and \emph{aggregated} (symbolic) semantics for choreographies, to reflect the GSOS-style semantics of processes, and to make the proof of the EPP theorem easier to carry out.

The \defn{local (symbolic) semantics} for choreographies is defined by the transitions in the top part of \Cref{fig:new-semantics}, where each transition is labelled by a
label in the following extended set of \defn{(local) labels}:
\begin{equation}
    \Llab =
    \setdef{
        \begin{array}{l@{\;,\;}l@{\;,\;}l}
            \PRsend rq      &
            \PRrecv pr      &
            \PRlsend rq     \;,\;
            \PRlrecv pr   
        \\
            \TRassign[\rclr{r}] xe & \TRelse[\rclr{r}] e &\TRthen[\rclr{r}] e 
        \end{array}
    }{
        \pid p , \pid q,\pid r \in \Pset, 
        e \in \Eset, 
        \albl \in \Lset,
        x \in \Varset
    }  
\end{equation}
where, we extend the definition of $\pnof\cdot$ to the new labels not in $\Glab$ as follows:
\[
    \pnof{\PRsend rq} = \pnof{\PRrecv pr} = \pnof{\PRlsend rq} = \pnof{\PRlrecv pr} = \set{\pid r}
\]

The \defn{aggregate (symbolic) semantics} for choreographies is 
defined by the transitions in the bottom part of \Cref{fig:new-semantics}, where each transition is labelled by a label in $\Glab$.

\begin{example}\label{ex:aggregate}
    Let $\cC$ be the choreography defined in \Cref{ex:chor}.

    In the aggregate semantics, the transition for the initial communication is decomposed in two local transitions, one for the sender and one for the receiver, where the client $\pid c$ performs a send action, while $\cas$ performs the corresponding receive action. 
    The aggregate semantics derives the communication transition as follows:

    $$
    \inferrule*[right=agg-com]{
        \inferrule*[right=$\pid c$ send]{ }{
            \cccom c{\creds}{\cas}{x} \seq\cCp \localsem[c]{\PRsend[\creds]{c}{\cas}} \cCp
        }
    \\
        \inferrule*[right=$\cas$ rec]{ }{
            \cccom c{\creds}{\cas}{x}\seq\cCp \localsem[\cas]{\PRrecv[x]{c}{\cas}} \cCp
        }
    }{
        \cC \aggsem{\TRcom[\creds]{c}{\cas}{x}} \cCp
    }
    $$

    From $\cCp=\cifte{\cas}{\valid(x)}{\cC[ok]}{\cC[ko]}$, we have two possible outcomes: the positive one, where the credentials are valid, and the negative one, where they are not. 
    In both cases, the transition is triggered by the process $\cas$ and involves only $\cas$ in the aggregate semantics.
    Below the derivation of the positive outcome, and the negative one is analogous (but with $\TRelse[\cas]{\valid(x)}$ instead of $\TRthen[\cas]{\valid(x)}$ and $\cC[ko]$ instead of $\cC[ok]$).
    
    $$
    \inferrule*[right=local] {      
        \inferrule*[right=$\pid{\cas}$ cond-then]{ }{
            \cifte{\cas}{\valid(x)}{\cC[ok]}{\cC[ko]}
            \localsem[\cas]{\TRthen[\cas]{\valid(x)}}
            \cC[ok]
        }
    }{ 
        \cifte{\cas}{\valid(x)}{\cC[ok]}{\cC[ko]}
        \localsem[\cas]{\TRthen[\cas]{\valid(x)}}
        \cC[ok]
    }
    $$

    Then, the continuation $\cC[ok]$ performs the selection with label
    $\lbl{ok}$, involving both the sender $\pid{\cas}$ and the receiver $\pid c$ of the label.
    This transition in the aggregate semantics is a \textsc{agg-sel}, which is derived by combining the local selection send of $\pid{\cas}$ and the local selection receive of $\pid c$ in the local semantics:

    $$
    \inferrule*[right=agg-sel]{
        \inferrule*[right=$\cas$ sendl]{  }{  
            \cC[ok]
            \localsem[\cas]{\PRlsend[ok]{\cas}{c}}
            \cCp[ok]
        }
        \ \
        \inferrule*[right=$\pid c$ recl]{  }{    
            \cC[ok]
            \localsem[c]{\PRlrecv[ok]{\cas}{c}}
            \cCp[ok]
        }
    }{
        \cC[ok]
        \aggsem{\TRsel[ok]{\cas}{c}}
        \ccsel{\cas}{c}{l} \seq \cCp[ok]
    }
    $$
    where we are letting $\cCp[ok]$ such that $\cC[ok]=\ccsel{\cas}{c}{l} \seq\cCp[ok]$.

    After another selection with label $\lbl{ok}$ between $\pid{\cas}$ and $\pid s$, we end up with a choreography 
    $$
    \cC[tail]=\cccom{s}{\token()}{c}{t} \seq \cccom{\cas}{\okval}{\logger}{\resvar} \seq \cnil 
    \qquad. 
    $$
    The reader familiar with the semantics of choreographies may recognise a scenario where both communications can be performed: the second one could be performed out of order by using the rule \textsc{delay} to delay the first communication, obtaining the choreography $\cCp[tail]=\cccom{s}{\token()}{c}{t} \seq \cnil$.
    The same translation in the aggregate semantics can be derived as follows:

    \[\scalebox{.76}{$
    \inferrule*[right=agg-com]{
        \inferrule*[right=$\cas$ delay] {
            \inferrule*[right=$\cas$ send]{  }{
                \cccom{\cas}{\okval}{\logger}{\resvar} \seq \cnil
                \localsem[\cas]{\PRsend[\okval]{\cas}{\logger}}
                \cnil
            }
        \ \ 
            \cas\notin\set{\pid s, \pid c} 
        } 
        {  \cC[tail]
            \localsem[\cas]{\PRsend[\okval]{\cas}{\logger}}
            \cccom{s}{\token()}{c}{t} \seq \cnil
        }
        \ \
        \inferrule*[right=$\logger$ delay] {
            \inferrule*[right=$\logger$ send]{  }{
                \cccom{\cas}{\okval}{\logger}{\resvar} \seq \cnil
                \localsem[\logger]{\PRrecv[\resvar]{\cas}{\logger}}
                \cnil
            }
        \ \ 
            \logger \notin \set{\pid s, \pid c} 
        } 
        {  
            \cC[tail]
            \localsem[\logger]{\PRrecv[\resvar]{\cas}{\logger}}
            \cccom{s}{\token()}{c}{t} \seq \cnil
        }
    }{
        \cC[tail]
        \aggsem{\TRcom[\okval] {\cas}{\logger}{\resvar}}
        \cccom{s}{\token()}{c}{t} \seq \cnil
    }
    $}\]

    The continuation of $\cC[ko]$ is analogous, but simpler.
\end{example}

A natural question is whether the new semantics for choreographies is equivalent to the traditional semantics, that is, whether the transitions of the new semantics can be simulated by the transitions of the traditional semantics, and vice versa.
The following result states that the two semantics are indeed equivalent, in the sense that for every transition of the new semantics there is a corresponding transition of the traditional semantics, and vice versa.

\begin{theorem}\label{thm:agg-trad}
    The traditional semantics and the aggregate semantics for choreographies are equivalent.
    That is, for every choreography $\cC$ and label $\mu\in\Glab$ we have that
    $\cC \tradsem{\mu} \cCp$ if and only if $\cC \aggsem{\mu} \cCp$.
\end{theorem}
\begin{proof}
    We prove both directions by case analysis on the possible transitions of the two semantics.

    We reason by induction on the derivation of a transition $\cC \tradsem{\mu} \cCp$ to show that a transition with the same source ($\cC$), label ($\mu$), and target ($\cCp$) is possible in the aggregate semantics.
    \begin{itemize}
        \item if the derivation of the transition is a single rule \textsc{rule}, then we have the following cases:
        \begin{itemize}
            \item $\textsc{rule}\in\set{\textsc{assignement},\textsc{cond-then},\textsc{cond-else}}$, then the same transition is possible in the aggregate semantics by the rule \textsc{local} that, in turn, is derivable by a rule in the local semantics with the corresponding label in $\set{\textsc{r assignement},\textsc{r cond-then},\textsc{r cond-else}}$.
            
            \item if $\textsc{rule}=\textsc{com}$, then the same transition with label $\mu=\TRcom pqx$ is possible in the aggregate semantics by the rule \textsc{agg-com} that, in turn, has two premises with transitions labeled by $\PRsend pq$ and $\PRrecv pq$ respectively,
            each of which is derivable by rules \textsc{p send} and \textsc{q recv} in the local semantics, respectively.

            \item if $\textsc{rule}=\textsc{sel}$, then the reasoning is similar to the previous case, by using the rule \textsc{agg-sel} instead of \textsc{agg-com}, and \textsc{p sendl} and \textsc{q recvl} instead of \textsc{p send} and \textsc{q recv}, respectively;
        \end{itemize}

        \item otherwise, the derivation of the transition $\cC \tradsem{\mu} \cCp$ starts with a rule $\textsc{rule}=\in\set{\textsc{delay},\textsc{delay-cond},\textsc{call}}$.
        In this case, we can immediately conclude by showing that a corresponding rule of the following form is admissible in the aggregate semantics.
        \[\begin{array}{c}
            \inferrule*[right=agg-delay] { \cC \aggsem{\mu} \cCp \\ \pn(\cI)\disj\pn(\mu)} 
            { \cI\seq \cC \aggsem{\mu}  \cI\seq \cCp}
        \hskip6em
            \inferrule*[right=agg-call]{
                \genProcDef \in \procC\\ 
                \cC \aggsem{\mu} \cCp
            }
            { \cX \aggsem{\mu}  \cCp }
        \\[5pt]
            \inferrule*[right=agg-delay-cond] { \cC[1] \aggsem{\mu}  \cCp[1] \\  \cC[2] \aggsem{\mu}  \cCp[2] \\ \pid{p}\notin \pn(\mu)}
            { \cgenifte \aggsem{\mu}  \cgeniftep }
        \end{array}\]
        
        We show the case for the rule \textsc{agg-delay}, the other cases are similar.
        We proceed by case analysis on the rule used in the derivation of the transition $\cC \aggsem{\mu} \cCp$.
        \begin{itemize}
            \item if there is a transition $\cC \aggsem{\mu} \cCp$ with $\mu \in \set{\TRassign xe , \ \TRthen[r]{e} , \TRelse[r]{e} \mid \ee \in \Eset, x \in \Varset}$ is derivable using the rule \textsc{local} with a premise a (local) transition $\cC \localsem[r]{\mu} \cCp$ derivable in the (local) semantics, then we can construct the following derivation of the transition $\cI\seq\cC \aggsem{\mu} \cI\seq\cCp$:
            $$
            \inferrule*[right=local] {
                \inferrule*[right=$\pid p$ delay] {
                    \cC \localsem[p]{\mu} \cCp
                    \ \ 
                    \qquad
                    \pid p \notin \pnof \mu=\set{\pid r}
                }{
                    \cI\seq \cC \localsem[p]{\mu}  \cI\seq \cCp
                }
            }{
                \cI\seq \cC \aggsem{\mu}  \cI\seq \cCp
            }
            $$
            \item 
            if there is a transition $\cC \aggsem{\TRcom pqx} \cCp$ derivable using the rule \textsc{agg-com} with premises the two (local) transitions $\cC \localsem[p]{\PRsend pq}  \cCp$ and $\cC \localsem[q]{\PRrecv pq}  \cCp$ derivable in the (local) semantics, then we can construct the following derivation of the transition $\cI\seq\cC \aggsem{\TRcom pqx} \cI\seq \cCp$:
            $$
            \inferrule*[right=agg-com] {
                \inferrule*[right=$\pid p$ delay] {
                    \cC \localsem[p]{\PRsend pq}  \cCp
                    \ \ 
                    \qquad
                    \pid p \notin \pnof\cI
                }{
                    \cI\seq \cC \localsem[p]{\PRsend pq}  \cI\seq \cCp
                }
                \ \
                \inferrule*[right=$\pid q$ delay] {
                    \cC \localsem[q]{\PRrecv pq}  \cCp
                    \ \ 
                    \qquad
                    \pid q \notin \pnof\cI
                }{
                    \cI\seq \cC \localsem[q]{\PRrecv pq}  \cI\seq \cCp
                }
            }{
                \cI\seq \cC \aggsem{\TRcom pqx}  \cI\seq \cCp
            }
            $$
            \item the case for the rule \textsc{agg-sel} is similar to the previous one, by using the corresponding rules and labels for selection instead of communication.
        \end{itemize}
    \end{itemize}
    
    The right-to-left direction is similar, taking each transition of the aggregate semantics and showing that it can be simulated by a transition of the traditional semantics.
\end{proof}

Because of this result, in the following we may use the traditional semantics and the aggregate semantics interchangeably, depending on which one is more convenient for the proof we are carrying out.
In particular, we may use notations $\tradsem{}$ and $\aggsem{}$ interchangeably.

\section{The Branch Preorder}\label{sec:EPPnew}

In this section, we introduce the main technical tool we use to prove the EPP theorem, that is, the \emph{branch preorder} relation $\mbb$ between choreographies and networks, capturing the fact that, due to branch prediction, the endpoint projection of a choreography may temporarily expose more local behaviours than the choreography itself.

It is defined by first defining the (local) relations $\mbb[p]$ between choreographies and networks at participant $\pid p$, thus using our newly defined local semantics of choreographies, which intuitively states that the local behaviour of participant $\pid p$ in the network $\nN$ is compatible with the local behaviour prescribed by the choreography $\cC$.
Here, compatibility is understood in a weak sense: the network may offer additional local choices, provided that they correspond to alternative branches of a conditional in the choreography.
Then, the (global) relation $\mbb$ is defined by requiring that $\mbb[p]$ holds for all participants $\pid p$.

\begin{definition}[Branch Preorder]\label{def:more_branch_cn}
    Let $\cC$ be a choreography, and let $\nN$ be a network.
    We say that $\cC$ has \defn{less branches} than $\nN$ \defn{with respect to $\pid p$} (denoted $\cC\mbb[p] \nN$) if 
    \begin{itemize}
        \item\label{def:more_branch_cn0} 
        \defn{\mbrel}:
        every process that is active in $\cC$ is also active in $\nN$.
        That is,
        \[\pid p \in \pnof{\cC} \implies \nN(\pid p) \neq \pnil\]

        \item\label{def:more_branch_cn1} 
        \defn{\mbGtoL}:
        if $\cC$ can make a local $\mu$-transition at $\pid p$ to $\cCp$, then the network $\nN$ can make the same transition to $\nNp$, and the derivatives are still related by $\mbb[p]$. That is, 
        \[\cC \localsem[p]{\mu} \cCp \implies \nN \netsem{\mu} \nNp \quand \cCp\mbb[p] \nNp\]
        
        \item\label{def:more_branch_cn2} 
        \defn{\mbLtoG}:
        if $\nN$ can make a local $\mu$-transition at $\pid p$ to $\nNp$, then either $\cC$ can make the same transition to $\cCp$, or the transition is the reception of a \bllabel that $\cC$ cannot perform, but the complementary \bllabel is available in $\cC$ and can be performed;
        in both cases, the derivatives obtained after performing the transition available to both $\cC$ and $\nN$ are still related by $\mbb[p]$.
        That is, 
        \[
            \nN \netsem{\mu} \nNp
            \text{ with }
            \pnof{\mu} = \set{\pid p}
         \implies 
            \begin{cases}
                \text{either}& \cC \localsem[p]{\mu} \cCp \mbb[p] \nNp 
            \\[4pt]
                \text{or}&
                \cC \not\localsem[p]{\mu} 
                \text{ with }
                \mu = \PRlrecv[$\albl$]{q}{p}
                \text{ , but}
            \\
                &
                \nN \netsem{\overline\mu} \nNpp
            \quand    
                \cC \localsem[p]{\overline\mu} \cCpp\mbb[p] \nNpp
                \text{ with } 
                \overline\mu=\PRlrecv[$\overline\albl$]{q}{p}
            \end{cases}
        \]
    \end{itemize} 
    We say that $\nN$ has \defn{more branches} than $\cC$ (denoted $\cC\mbb \nN$) if $\cC\mbb[p] \nN$  for all $\pid p\in\Pset$.
\end{definition}

Notice that Condition~\ref{def:more_branch_cn1} requires every local transition of $\cC$ at $\pid p$ to be matched by an identical transition of $\nN$, while  Condition~\ref{def:more_branch_cn2} allows $\nN$ to perform certain local transitions that are not immediately available in $\cC$: either the choreography can perform the same transition, or the transition corresponds to the reception of a \bllabel that originates from branch prediction. Then the choreography is required to offer the complementary reception, and the derivatives
obtained after receiving the complementary labels must again be related
by $\mbb[p]$.
As a consequence, $\mbb$ does not require the choreography and the
network to expose exactly the same set of local transitions.
Rather, it allows the network to anticipate branch-dependent behaviour,
while preserving compatibility with the future evolution prescribed by
the choreography.

\begin{example}
    Consider the choreography and network in
    \Cref{eq:more-branch-example}.

    \begin{equation}\label{eq:more-branch-example}
        \begin{array}{l@{\ =\ }l}
        \cC
        &
        \cifte pe{
            \left(\ccsel pq{$\lleft$} \seq \cnil\right)
        }{
            \left(\ccsel pq{$\lright$} \seq \cnil\right)
        }
    \\[8pt]
        \nN
        &
        \pdo[p]{
            \pifte e
            {\left(\psendl[\lleft]{q} \pseq \pnil\right)}
            {\left(\psendl[\lright]{q} \pseq \pnil\right)}
        }
        \ppar
        \pdo[q]{
            \plabels{p}{
                \lleft : \pnil,
                \;
                \lright : \pnil
            }
        }.
        \end{array}
    \end{equation}

    Before participant $\pid p$ resolves the conditional,
    every local transition of $\cC$ is matched by a corresponding
    local transition of $\nN$.
    Assume now that both the choreography and the network select the left branch.
    The choreography evolves to
    $\ccsel pq{$\lleft$} \seq \cnil$,
    while the network evolves to a state where participant
    $\pid q$ still offers both receptions
    $\lleft$ and $\lright$.
    Consequently, the local reception
    $\PRlrecv[$\lright$]{p}{q}$
    is available in the network but not in the choreography.
    Nevertheless, the resulting systems remain related by $\mbb$.
    Indeed, the additional reception corresponds exactly to the
    unselected branch of the original conditional and therefore
    represents a behaviour predicted by the endpoint projection.
    This example illustrates the purpose of
    Condition~\ref{def:more_branch_cn2}, which allows projected
    networks to expose branch-dependent local behaviours that are
    no longer present in the choreography while preserving semantic
    compatibility.
\end{example}

\subsection{Stratified Branch Preorder}

Before proving the main result, we introduce a stratification $\mbbn[n]{}$ of the more-branches relation using the standard construction from fixed-point theory to characterise bisimilarity via successive approximants~\cite[Chapter~2]{S11}.
This stratification yields an inductive characterisation of $\mbb$ (see\Cref{lem:mbb_strat}) allowing us to prove our results by induction (on the stratification of $\mbb$), which is technically more convenient than proving it directly on $\mbb$ using coinduction.

\begin{definition}\label{def:mbb_strat}
    The \defn{stratification} $\mbbn[\omega]{r}$ of the more-branches relation is defined inductively:
    \begin{enumerate}
        \item\label{mbb_strat:0} $\cC \mbbn[0]{r} \nN$ for all $\cC$ and $\nN$;
        \item $\cC \mbbn[n+1]{r} \nN$ if the following hold%
        \footnote{These conditions are the bounded version of the ones of \cref{def:more_branch_cn}, which can be consulted for a more discursive definition.}%
        :
        \begin{itemize}\label{mbb_strat:Step}
            \item\label{def:mbb_strat0} 
            \defn{\mbrel}:
            if $\pid p \in \pnof{\cC}$, then $\nN(\pid p) \neq \pnil$ ;

            \item\label{def:mbb_strat1} 
            \defn{\mbGtoL}:
            if $\cC \localsem[r]{\mu} \cCp$, then $\nN \netsem{\mu} \nNp$ and $\cCp\mbbn{r} \nNp$ ;

            \item\label{def:mbb_strat2} 
            \defn{\mbLtoG}:
            if $\nN \netsem{\mu} \nNp$ with $\pnof{\mu} = \set{\pid r}$, 
            then either $\cC \localsem[r]{\mu} \cCp \mbbn{r} \nNp $, or $\cC$ cannot perform $\mu=\PRlrecv[$\albl$]{q}{r}$ but both $\cC$ and $\nN$ can perform $\overline\mu=\PRlrecv[$\overline\albl$]{q}{r}$, and $\cC \localsem[r]{\overline\mu} \cCpp\mbb[r] \nNpp\conetsem{\overline\mu} \nN$.
        \end{itemize}

        \item\label{mbb_strat:Omega} $\cC \mbbn[\omega]{r} \nN$ if $\cC \mbbn{r} \nN$ for all $n \in \mathbb N$.
        
    \end{enumerate} 
    We say that $\nN$ has \defn{more branches up-to $n$} than $\cC$ (denoted $\cC\mbbn{} \nN$) if $\cC\mbbn{p} \nN$  for all $p\in\Pset$.
\end{definition}

In the following lemmas we prove that $\mbb[r]$ and $\mbbn[\omega]{r}$ are the same relation, using the standard proof technique for showing that a bisimulation relation coincides with its stratification .
The proof in \cite[Chapter~2]{S11} of this result is based on the assumption that the LTSs under consideration are finitely branching.
In our setting we are restricting the scope of the proof to deterministic LTSs, in the sense that given a state $\cC$ \resp{$\nN$} and a label $\mu$, if a transition is defined, then it is unique, since both the one induced by local semantics and network semantics are deterministic.  
The proof of \Cref{lem:mbb_strat} can be extended to finite branching LTS, as in~\cite[Chapter~2]{S11}. 

\begin{lemma}\label{lem:decreasing}
    For every $n \in \mathbb N$, 
    the relation $\mbbn[n+1]{p}$ is smaller of $\mbbn[n]{p}$, that is, $\mbbn[n+1]{p} \subeqw \mbbn[n]{p}$.
\end{lemma}
\begin{proof}
    The proof proceeds by induction on $n$. The base case is immediate, since $\mbbn[0]{p}$ is the total relation. The inductive step follows directly from the definition and the induction hypothesis.
\end{proof}

\begin{lemma}\label{lem:mbb_strat}
    The relations $\mbb[r]$ and $\mbbn[\omega]{r}$ coincide.
\end{lemma}
\begin{proof}
    We prove both inclusions.
    \begin{itemize}
        \item[$(\subseteq)$]
        We show by induction on $n$ that $\mbb[r] \subeqw \mbbn[n]{r}$ for all $n \geq 0$.
        \begin{itemize}
            \item Case base: the relation $\mbb[r]$ satisfies the same condition as $\mbbn[0]{r}$ plus additional requirements, hence $\mbb[r] \subeqw \mbbn[0]{r}$.

            \item Inductive step: 
            Since $\mbb[r]$ satisfies the clauses of \Cref{def:more_branch_cn}, if $\cC \mbb[r] \nN$, then the induction hypothesis $\cC \mbbn[n]{r} \nN$ immediately yields $\cC \mbbn[n+1]{r} \nN$.
        \end{itemize}

        \item[$(\supseteq)$]
        We show that $\mbbn[\omega]{r}$ satisfies the clauses of \Cref{def:more_branch_cn}. Assume $\cC \mbbn[\omega]{r} \nN$, we show that $\mbbn[\omega]{r}$ satisfies the three conditions of \Cref{def:more_branch_cn}.
        \begin{itemize}
            \item[$\bullet$]\mbrel: trivial by definition of $\mbbn[\omega]{r}$.
            
            \item[$\bullet$]\mbGtoL:
            suppose $\cC \localsem[r]{\mu} \cCp$. For every $n \geq 1$, since $\cC \mbbn[n]{r} \nN$, there exists $\nNp$ such that
            $
                \nN \netsem{\mu} \nNp
                \quad\text{and}\quad
                \cCp \mbbn[n-1]{r} \nNp.
            $
            By determinism of the LTS, $\nNp$ is unique; hence the same $\nNp$ works for all $n$, yielding $\cCp \mbbn[\omega]{r} \nNp$.

            \item[$\bullet$]\mbLtoG:
            suppose $\nN \netsem{\mu} \nNp$ with $\pnof{\mu} = \set{\pid r}$. For every $n \geq 1$, since $\cC \mbbn[n]{r} \nN$, either:
            \begin{enumerate}
                \item
                $\cC \localsem[r]{\mu} \cCp$
                and
                $\cCp \mbbn[n-1]{r} \nNp$; or

                \item
                $\mu = \PRlrecv[\albl]{q}{r}$,
                $\cC$ has no $\mu$-transition,
                and there exist $\cCpp,\nNpp$ such that
                $
                    \cC \localsem[r]{\PRlrecv[$\overline{\albl}$]{q}{r}}
                        \cCpp[]$ ,
                $
                    \nN 
                    \netsem{
                        \PRlrecv[$\overline{\albl}$]{q}{r}
                    }
                        \nNpp[] ,
                $
                and
                $
                    \cCpp \mbbn[n-1]{r} \nNpp.
                $
            \end{enumerate}
        \end{itemize}
        Determinism guarantees uniqueness of the corresponding successor states, 
        and since the two cases are mutually exclusive, exactly one of the two cases holds for all $n$, hence, the same successors satisfy the
        required relations for all $n$. Therefore, $\mbbn[\omega]{r}$ satisfies \Cref{def:more_branch_cn}.
    \end{itemize}
\end{proof}

\subsection{A choreography $\mbb$-smaller than a network has local progress}\label{sec:MBandEPP}

A key property that we use in our new proof of the EPP theorem is that we can use the local semantics of choreographies to check the local behaviour of the network at each participant.
Before stating this result, we need an intermediate result that ensures that we can check whether $\cC\mbb\nN$ by only checking the local behaviour of $\cC$ and $\nN$ at each participant, without having to check the global behaviour of $\cC$ and $\nN$.

\begin{theorem}\label{thm:mb_locality}
    $\cC \mbb \nN$ if and only if $\cC \mbb \pdo{\nN(\pid p)}$ for all $\pid p \in \Pset$.  
\end{theorem}
\begin{proof}
    By \Cref{lem:mbb_strat} and by definition of $\mbb$,
    it suffices to check that $\cC \mbbn[n]{p} \nN$ holds if and only if $\cC \mbbn[n]{p} \pdo{\nN(\pid p)}$ holds, for all $\pid p \in \Pset$ and for all $n \geq 0$.
    Then we conclude by \Cref{lem:mb_net}.\ref{lem:mb3}.
\end{proof}

\begin{lemma}\label{lem:mb_wf_merge}
    Let \emph{$\cC = \cifte pe{\cC[1]}{\cC[2]}$} be a choreography, $\nN$ a network such that $\cC\mbb[r] \nN$ for any $\pid r$ such that $\pid r\neq \pid p$, and $i, j \in \set{1,2}$ with $i \neq j$.
    If $\cC[i] \localsem[r]{\mu} \cCp[i]$, then 
    \begin{itemize}
        \item either
        $\cC[j] \localsem[r]{\mu} \cCp[j]$;
        
        \item 
        or 
        $\mu = \PRlrecv[$\albl$]{q}{r}$
        and 
        $\cC[j] \notlocalsem[r]{\mu}$, but $\cC[j] \localsem[r]{\PRlrecv[$\overline{\albl}$]{q}{r}} \cCp[j]$.
    \end{itemize}

\end{lemma}
\begin{proof}
    We proceed by contradiction. Suppose that $\cC[j]$ admits no transition labelled 
    by $\mu$ at $\pid{r}$, and that if $\mu = \PRlrecv[$\albl$]{q}{r}$ then $\cC[j]$ 
    also admits no transition labelled by $\PRlrecv[$\overline{\albl}$]{q}{r}$ at $\pid{r}$. By definition of local semantics (see \Cref{fig:new-semantics}) 
    we have that $\cifte pe{\cC[1]}{\cC[2]}$ is stuck at $\pid r$.
    By \Cref{def:more_branch_cn}.\ref{def:more_branch_cn0} sice $\pid r \in \pnof{\cifte pe{\cC[1]}{\cC[2]}}$ then $\nN(\pid r) \neq \pnil$, 
    hence by \Cref{lem:local_progress_net} and by definition of network semantics (see \Cref{fig:network:semantics}) there is $\tilde{\mu}$ such that $\pid r \in \pnof{\tilde{\mu}}$ and 
    $\nN \netsem{\tilde{\mu}} \nNp$ contradicting \Cref{def:more_branch_cn}.\ref{def:more_branch_cn1}.
\end{proof}

To prove that the branch preorder coincides with bisimilarity, we make use of the stratified construction of $\mbb$, which facilitates a more straightforward proof of the result.
We first prove a \emph{local progress} property for choreographies that are $\mbbn[n]{p}$-smaller than a network.
\begin{proposition}[Local progress given a network]\label{prop:local_progress}
    Let $\cC$ be a choreography and \emph{$\pid p \in \pnof{\cC}$}.
    If there is a network $\nN$ such that $\cC \mbbn[n]{p} \nN$, then $\cC \localsem[p]{\mu} \cCp$ for some label $\mu$ and some choreography $\cCp$.
\end{proposition}
\begin{proof}
 By \Cref{def:mbb_strat}.\ref{def:mbb_strat0} $\nN(\pid p) \neq \pnil$ then by \Cref{lem:local_progress_net} there are $\mu'$ and $\nNp$ such that $\nN \netsem{\mu'} \nNp$.
 We conclude by \Cref{def:mbb_strat}.\ref{def:mbb_strat1}.
\end{proof}
Note that the assumption $\cC \mbbn[n]{p} \nN$ is crucial, as the fact that $\pid p$ is a process in $\pnof{\cC}$ alone is insufficient
to establish the result.
For instance, in $\cifte qe{\left(\cloc pxf \seq \cnil\right)}{\cnil}$, the name $\pid p$ occurs in the choreography, but the conditional is controlled by $\pid q$, and there is no local transition available at $\pid p$ before the choice is resolved.

\begin{corollary}\label{cor:well_guardedness}
    Let \emph{$\cC = \cifte pe{\cC[1]}{\cC[2]}$} be a choreography and $\nN$ a network such that $\cC \mbb[r] \nN$ for a $\pid r\neq \pid p$. Then \emph{$\pid r \in \pnof{\cC[1]}$} if and only if \emph{$\pid r \in \pnof{\cC[2]}$}.
\end{corollary}
\begin{proof}
    Follows directly by \Cref{prop:local_progress} and \Cref{lem:mb_wf_merge}.
\end{proof}

We can now, thanks to the following technical lemmas, prove the main theorem of this section stating that the branch preorder $\mbb$ coincides with bisimilarity.
The proofs are quite technical and rely on the definition of $\mbbn{}$ and on the properties of local semantics, requiring tedius case analysis and inductions.
To avoid distracting the reader from the main ideas of the proof, we have moved these technical lemmas and their proofs in \Cref{app:proofs}.

\begin{restatable}{lemma}{lemmaMBbasic}\label{lem:mb}
    Let $\cC$, $\cC[1]$, and $\cC[2]$ be choreographies, 
    $\nN$ a network.
    The following hold for any $n \geq 0$ and $\pid p, \pid r \in \Pset$ with $\pid p \neq \pid r$:
    \begin{enumerate}
        \item\label{lem:mb1} If $\cC \mbbn{r} \nN$ and $\nN \netsem{\mu} \nNp$ with $\pnof{\mu} = \set{\pid p}$,
        then $\cC \mbbn{r} \nNp$.

        \item\label{lem:mb2}
        \emph{$\cifte pe{\cC[1]}{\cC[2]} \mbbn{r} \nN$} if and only if
        $\cC[i] \mbbn{r} \nN$ for both $i \in\set{1,2}$.

        \item\label{lem:mb3} $\cI \seq \cC \mbbn{r} \nN$, for \emph{$\pid r \notin \pnof{\cI}$} if and only if
        $\cC \mbbn{r} \nN$.
    \end{enumerate}
\end{restatable}

\begin{restatable}{lemma}{lemmaUgly}\label{lem:ugly}
    Let $\cC$, $\cC[1]$, and $\cC[2]$ be chorographies, $\cX$ a choreography variable, and $\nN$ a network.
    The following hold for any $n \geq 0$, for any $\albl \in \Lset$, and for any $\pid p, \pid q, \pid r \in \Pset$ with $\pid p\neq \pid r \neq \pid q,  $:
    \begin{enumerate}
        \item\label{lem:mb4} 
        $\cC \mbbn[n]{r} \pdo[r]{\pR}$
        ~if and only if~
        $\cI \seq \cC \mbbn[n+1]{r} \pdo[r]{\palpha \pseq \pR}$
        for any instruction $\cI\in\set{\cgenloc, \cgencom, \cgensel}$ and a compatible $\palpha$ such that $\pid r \notin \palpha$;

        \item\label{lem:mb4bis}
        $\cC \mbbn[n]{r} \pdo[r]{\pR[{\albl}]}$
        ~if and only if~
        $\ccsel pr{$\albl$} \seq \cC 
        \mbbn[n+1]{r}
        \pdo[r]{\plabel[\pR]p{\albl}{\Lset}}$;

        \item \label{lem:mb5} 
        $\cC[i] \mbbn{p} \pdo[p]{\pP[i]}$
        ~if and only if~
        \emph{$\cifte pe{\cC[1]}{\cC[2]} \mbbn[n+1]{p}
        \pdo[p]{\pifte e{\pP[1]}{\pP[2]}}$}
        for any $i\in\set{1,2}$;

        \item\label{lem:mb_net4}\label{lem:prefix1} 
        if ~$\cI \seq \cC \mbbn{p} \pdo[p]{\pP}$ and $\pid p \in \pnof{\cI}$, 
        then 
        either 
        $\pP = \palpha \pseq \pP[1]$ 
        with $\palpha$ compatible with $\cI$ and $\pid p \notin \palpha$, 
        or 
        $\cI = \cgensel $ and $\pP= \plabel p{\albl}{\Lset}$;
        
        \item\label{lem:mb_net5}\label{lem:prefix2} 
        if ~\emph{$\cifte pe{\cC[1]}{\cC[2]} \mbbn{p} \pdo[p]{\pP}$}, then \emph{$\pP = \pifte e{\pP[1]}{\pP[2]}$};
        
        \item\label{lem:mb_net6} 
        if ~$\cX \mbbn[n]{} \nN$ and $\genProcDef \in \procC$, then $\cC \mbbn{} \nN$.
    \end{enumerate}
\end{restatable}

\begin{restatable}{theorem}{thmStratifiedSimu}\label{thm:main}
    Let $\cC$ be a choreography, and let $\nN$ be a network.
    Then, for all $n \geq 0$,
    $$\cC \mbbn{} \nN \quad \implies \quad \cC \bisimn \nN
    \mydot
    $$
\end{restatable}

As a consequence of \Cref{thm:main} and by definition of $\mbb$ and $\bisim$, the proof of the EPP theorem becomes an immediate corolloary

\begin{corollary}\label{cor:main}
    Let $\cC$ be a choreography, and let $\nN$ be a network.
    If $\cC \mbb \nN$, then $\cC \bisim \nN$.
\end{corollary}

\section{A New Proof of the EPP Theorem}\label{sec:epp_new}

In this section, we use the results about the branch preorder to give a new proof of the EPP theorem, which states that the endpoint projection of a choreography is bisimilar to the original choreography.
For this purpose, we reformulate the result of the previous section by replacing the assumption $\cC \mbb \nN$ with the assumption that $\cC$ is projectable and $\nN$ being the endpoint projection of $\cC$.

\begin{lemma}\label{lem:epp_cond1}
    If $\cC$ is projectable at $\pid p$ for \emph{$\pid p \in \pnof{\cC}$} and $\cC \localsem[p]{\mu} \cCp$ then  
    $\pdo{\proj[p]{\cC}}\netsem{\mu} \pdo{\proj[p]{\cCp}}$.
\end{lemma}
\begin{proof}
    By induction on the derivation of the transition $\cC \localsem[p]{\mu} \cCp$. 
    \begin{itemize}
        \item If the last rule is \cthen[p], \celse[p], \ccassign[p], \csend[p], \crec[p], \csendl[p], \crecl[p]
        the claim follows directly from definition of endpoint projection (see \Cref{fig:projection}) and network semantics. 
        \item If the last rule is a \cdelay[p] then $\cC = \cI \seq \cC[1]$ with $\pid p \notin \pnof{\cI}$. 
        Since $\proj[p]{\cC} = \proj[p]{\cC[1]}$ we conclude by induction hypothesis.
        \item If the last rule is a \delaycond[p] then $\cC = \cifte qe{\cC[1]}{\cC[2]}$ with $\pid q \neq \pid p$ and 
        $\proj[p]{\cC} = \proj[p]{\cC[1]} \sqcup \proj[p]{\cC[2]}$. We conclude by \Cref{lem:merge_and_local}.\ref{lem:merge_local_sem1}.
        \item If the last rule is a \delaypredictl[p] (resp. \delaypredictr[p]) we conclude by \Cref{lem:merge_and_local}.\ref{lem:merge_local_sem2}
        (resp \Cref{lem:merge_and_local}.\ref{lem:merge_local_sem2} and {the commutativity of merge in proposition~\ref{prop:merge}}).
        \item If the last rule is a \ccall[p] we conclude by induction hypothesis.
    \end{itemize}
\end{proof}

The local progress, now becomes a consequence of the projectability of the choreography.

\begin{proposition}[Local progress for projectable choreographies]\label{prop:projectable_local_prog}
    If $\cC$ is projectable at \emph{$\pid p \in \pnof{\cC}$}, 
    then there is $\mu$ such that $\cC \localsem[p]{\mu} \cCp$ for some $\cCp$.
\end{proposition}
\begin{proof}
    By structural induction on $\cC$.
    \begin{itemize}
        \item If $\cC = \cnil$, $\pid p \notin \pnof{\cC}$: the statement holds trivially. 
        \item If $\cC = \cI \seq \cC[1]$ there are two cases $\pid p \in \pnof{\cI}$ or otherwise. In the first case the claim follows from definition of local semantics. 
        If $\pid p \in \pnof{\cI}$ then $\pid p \in \pnof{\cC[1]}$. Since $\cC[1]$ is projectable by induction hypothesis there are $\mu$ and $\cCp[1]$ such that
        $\cC[1] \localsem[p]{\mu} \cCp[1]$. We infer $\cC \localsem[p]{\mu} \cCp$ by using a rule \cdelay[p].
        \item If $ \cC =\cifte pe{\cC[1]}{\cC[2]}$ by definition of projection and merge (see \Cref{fig:projection}) $\pid p \in \pnof{\cC[i]}$ if and only if 
        $\pid p \in \pnof{\cC}$, for $i= 1,2$. By induction hypothesis there are $\mu_1$, $\mu_2$, $\cCp[1]$ and $\cCp[2]$ such that $\cC[i] \localsem[p]{\mu} \cCp[i]$
        for $i=1,2$. By \Cref{lem:epp_cond1} $\pdo[p]{\proj[p]{\cC[i]}} \localsem[p]{\mu} \pdo[p]{\proj[p]{\cCp[i]}}$ and by \Cref{lem:merge_and_local}.\ref{lem:merge_local_sem5} 
        either $\mu_1 = \mu_2$ or $\mu_i = \PRlrecv[$\albl$]{q}{p}$ and $\mu_j = \PRlrecv[$\overline{\albl}$]{q}{p}$, for $i \neq j$, $i.j \in \set{1,2}$.
        If $\mu_1 = \mu_2$ we infer $\cC \localsem[p]{\mu} \cCp$ by using a rule \cdelay[p], otherwise we use \delaypredictl[p].
        \item If $\cC = \cX[i]$, with $\genProcDef[i] = \cC[i]$, we conclude by
        guardedness and induction hypothesis.
        Indeed, the previous cases establish the claim for finite
        choreographies, namely choreographies containing no choreography
        variables. By guardedness, every participant in
        $\pnof{\cX[i]}$ also belongs to $\apnof{\cC[i]}$.
        Hence, $\cC[i]$ can perform a local transition at $\pid p$
        if and only if the finite choreography obtained from $\cC[i]$
        by replacing every choreography variable with $\cnil$
        can perform a local transition at $\pid p$.
        Observe that this finite choreography is still projectable.
        Indeed, this follows directly from the definition of projection (see \Cref{fig:projection}) and
        merge. Replacing choreography variables with $\cnil$ can only affect
        projections of the form $\proj[p]{\cX}$, but by definition of merge
        we have $\cX \sqcup \pQ$ defined only if $\pQ = \cX$, and similarly
        $\pnil \sqcup \pQ$ is defined only if $\pQ = \pnil$.
        Hence, replacing every choreography variable with $\cnil$ preserves the
        definedness of all merge operations.
        We can therefore apply the finite case proved above and conclude that a
        local transition at $\pid p$ exists.
    \end{itemize}
\end{proof}

\begin{theorem}\label{thm:cor_mbstrat_p_proj}
    Let $\cC$ be a choreography.
    If $\cC$ is projectable,
    then
    $\cC \mbb[p] \pdo[p]{\proj[p]{\cC}}$.
\end{theorem}
\begin{proof}
    By \Cref{lem:mbb_strat}, if is sufficient to show that, for all $n \geq 0$, $ \cC \mbbn[n]{p} \pdo[p]{\proj[p]{\cC}}$.
    We proceed by induction on $n$.

    \begin{itemize}
        \item Base case: trivial since $\mbbn[0]{p}$ is the total relation. 

        \item Inductive step:
        Assume that, for every choreography $\cC[0]$ and every $\pid p \in \pnof{\cC[0]}$, if $\cC[0]$ is projectable at $\pid p$ then $\cC[0] \mbbn[n]{p} \pdo[p]{\proj[p]{\cC[0]}}$.
        We show that $\cC \mbbn[n+1]{p} \pdo[p]{\proj[p]{\cC}}$ by checking the conditions from \Cref{def:mbb_strat}.
        \begin{enumerate}
            \item Holds by induction hypothesis.
            \item If $\cC \localsem[p]{\mu} \cCp$, then we have that 
            $
            \pdo[p]{\proj[p]{\cC}}
            \netsem{\mu}
            \pdo[p]{\proj[p]{\cCp}}
            $ by  \Cref{lem:epp_cond1}.
            Hence,
            $
            \cCp \mbbn[n]{p}
            \pdo[p]{\proj[p]{\cCp}}
            $
            by the induction hypothesis.

            \item 
            If $\pdo[p]{\proj[p]{\cC}} \netsem{\mu} \nNp$
            with $\pnof{\mu} = \set{\pid p}$, then
            we must have $\pid p \in \pnof{\cC}$ 
            by definition of projection
            because ${\proj[p]{\cC}} \neq \pnil$
            Hence, by \Cref{prop:projectable_local_prog}, there exist $\mu'$ and $\cCp$ such that
            $\cC \localsem[p]{\mu'} \cCp$.
            Moreover, by \Cref{lem:epp_cond1} we also have $\pdo[p]{\proj[p]{\cC}} \netsem{\mu'} \pdo[p]{\proj[p]{\cCp}}$.
            According to the definition of semantics for networks,
            the transitions
            $\pdo[p]{\proj[p]{\cC}} \netsem{\mu} \nNp$
            and
            $\pdo[p]{\proj[p]{\cC}} \netsem{\mu'} \pdo[p]{\proj[p]{\cCp}}$
            can coexist only in one of the following situations:
            \begin{itemize}
                \item $\mu = \mu'$:
                In this case, the determinism of the network semantics implies $\nNp = \pdo[p]{\proj[p]{\cCp}}$. Since $\cCp$ is projectable at $\pid p$, we can apply the induction hypothesis and deduce that $\cCp \mbbn[n]{p} \pdo[p]{\proj[p]{\cCp}}$.
                
                \item $\mu,\mu' \in \set{\PRthen[p]{e},\PRelse[p]{e}}$:
                In this case, we consider the finite derivation tree witnessing $\cC \localsem[p]{\mu'} \cCp$.
                Since $\mu'$ is either $\PRthen[p]{e}$ or $\PRelse[p]{e}$, then the derivation must begin with a rule \cthen[p] or \celse[p], respectively.
                If $\cCpp$ be the choreography obtained by replacing this initial rule with the complementary one, then $\cC \localsem[p]{\mu} \cCpp$.
                Therefore, by \Cref{lem:epp_cond1}, we have 
                $
                \pdo[p]{\proj[p]{\cC}}
                \netsem{\mu}
                \pdo[p]{\proj[p]{\cCpp}}
                $
                and by the determinism of semantics of network, it follows that $\nNp = \pdo[p]{\proj[p]{\cCpp}}$.
                Since $\cCpp$ is projectable, we conclude by induction hypothesis that $\cCpp \mbbn[n]{p} \nNp$.

                \item $\mu,\mu' \in \set{\PRlrecv[$\albl$]{q}{p},\PRlrecv[$\overline{\albl}$]{q}{p}}$:
                In this case, we have exactly the label-receive exception allowed by the first clause \Cref{def:mbb_strat}.\ref{def:mbb_strat2}.
                Therefore, it suffuces to prove that also the second clause of \Cref{def:mbb_strat}.\ref{def:mbb_strat2} is satisfied to conclude.

                This follows by the fact that, if  $\mu = \PRlrecv[$\albl$]{q}{p}$ and $\mu' = \PRlrecv[$\overline{\albl}$]{q}{p}$, then  \Cref{lem:epp_cond1} applies, giving us
                $
                \pdo[p]{\proj[p]{\cC}}
                \netsem{\mu'}
                \pdo[p]{\proj[p]{\cCpp}}
                $.
                Since $\cCpp$ is projectable, we conclude by induction hypothesis that
                $\cCpp \mbbn[n]{p} \pdo[p]{\proj[p]{\cCpp}}$.

            \end{itemize}
        \end{enumerate}
    \end{itemize}
\end{proof}

We have shown that, for every participant, the choreography is $\mbb$-smaller than its projection onto that participant.
To conclude the proof of endpoint projection, it remains to lift this local result to the whole projected network.
For this purpose, we need to establish some additional properties of the branching preorder, which are stated in the following lemma, showing that, 
if we want to show that a choreography $\cC$ is $\mbb$-smaller than a network $\nN$, then it is sufficient to show that $\cC$ is $\mbb$-smaller than each atomic network $\pdo{\nN(\pid p)}$ in forming $\nN$.

\begin{lemma}\label{lem:mb_net}\label{lem:prefix}
    For all $n \geq 0$ and $\pid p \in \Pset$, $\pid p$,  for all $\cC \in \Chor$:
    \begin{enumerate}
        \item\label{lem:mb_net1} if $\cC \mbbn[n+1]{p} \nN$ then $\pid p \in \pnof{\nN}$ if and only if $\pid p \in \pnof{\cC}$;
        \item\label{lem:mb_net2} if $\cC \mbbn{p} \nN$ then $\cC \mbbn{p} \nN \ppar \nM$;
        \item\label{lem:mb_net3} $\cC \mbbn{p} \nN$ if and only if $\cC \mbbn{p} \pdo{\nN(\pid p)}$.
    \end{enumerate}
\end{lemma}
\begin{proof} 
    We prove each item separately.
    \begin{enumerate}
        \item Follows directly by \Cref{def:mbb_strat} and \Cref{prop:local_progress}.
        \item It is sufficient to observe that since $\pid p \notin \dom{\nM}$, 
        for all $\mu$ such that $\pnof{\mu} = \set{\pid r}$ then 
        $\nN \netsem{\mu} \nNp$ if and only if $ \nN \ppar \nM \netsem{\mu} \nNp \ppar \nM$.
        \item Right to left implication follows form \Cref{lem:mb_net}.\ref{lem:mb_net2}.
        We show left to right implication by induction on $n$.
        \paragraph*{Base case $n=0$.}
        Immediate since $\mbbn[0]{p}$ is the total relation. 
        
        \paragraph*{Inductive step.}
        Assume the statement is true for $n$, we prove it for $n+1$. 
        In particular, we prove that for all $\pid p \in \Pset$,  $\cC \mbbn[n+1]{p} {\nN}$ implies $\cC \mbbn[n+1]{p} \pdo{\nN(\pid p)}$.
        To conclude, it suffices to check the properties of $\mbbn[n+1]{p}$:
        \begin{itemize}
            \item \mbrel is trivial;
            
            \item \mbGtoL: 
            if $\cC \localsem[p]{\mu} \cCp$ then $\nN \netsem{\mu} \nNp$ and $\pnof{\mu} = \set{\pid p}$, hence, by \Cref{lem:net_sem} $\pdo{\nN(\pid p)} \netsem{\mu} \pdo{\nNp(\pid p)}$.
            By \Cref{def:mbb_strat}, $\cCp \mbbn[n]{p} \nNp$ hence by induction hypothesis on $n$ we obtain $\cCp \mbbn[n]{p} \pdo{\nNp(\pid p)}$;
            
            \item \mbLtoG: 
            if $\pdo{\nN(\pid p)} \netsem{\mu} \pdo{\nNp(\pid p)}$ with $\pnof \mu = \set{\pid p}$,  then by \Cref{lem:net_sem} $\nN \netsem{\mu} \nNp$. 
            Then,  either 
            \begin{itemize}
                \item $\cC \localsem[p]{\mu} \cCp$  and $\cCp \mbbn{p} \nNp$, in which case we conclude by induction hypothesis on $n$, or
                \item $\cC \localsem[p]{\overline{\mu}} \cCpp$ and $\nN \netsem{\overline{\mu}} \nNpp$ and $\cCpp \mbbn{p} \nNpp$, with $\pnof{\overline{\mu}} = \set{\pid p}$.
                By \Cref{lem:net_sem} we deduce $\pdo{\nN(\pid p)} \netsem{\overline{\mu}} \pdo{\nNpp(\pid p)}$ and we conclude
                applying induction hypothesis to $\cCpp \mbbn{p} \nNpp$.
            \end{itemize}
        \end{itemize}
    \end{enumerate}
\end{proof}

We can now prove the EPP theorem as a direct consequence of the new results we have proved so far.
\EPP*
\begin{proof}\label{thm:epp}[New Proof]
    The proof follows directly by the following chain of implications:
    \begin{itemize}
        \item By \Cref{thm:cor_mbstrat_p_proj}, \Cref{lem:mb_net}.\ref{lem:mb_net3} and by definition of endpoint projection, if $\cC$ is projectable then $\cC \mbb[] \proj[]{\cC}$.
        \item By \Cref{cor:main}, if $\cC \mbb[] \proj[]{\cC}$ then $\cC \bisim \proj[]{\cC}$.
    \end{itemize}
\end{proof}

\begin{remark}
    A key advantage of our framework is that the proof of endpoint projection can be reduced to a collection of local behavioural checks.
    Indeed, by \Cref{lem:mb_net}.\ref{lem:mb_net3}, proving that a choreography has less branches than its projection amounts to proving the corresponding local property for each participant separately.
    This avoids the need to reason directly about the behaviour of the entire projected network and considerably simplifies the proof structure.
\end{remark}

\section{Extending our Results to the Stateful Setting}\label{sec:stores}

In this section, we show how to extend our results to the standard semantics of choreographies, which usually also take into account the state of the processes, by adding a memory store to the semantics to keep track of the values of the variables at each process.
We show how to extend the definitions of the previous sections to the setting with stores, and we show that all the results presented in the previous sections still hold in this extended setting.
In particular, we reconstruct the standard stateful semantics of choreographies by using a single rule factoring the evolution of the choreography and of the store using the symbolic semantics of choreographies and a new semantics for stores.

We recall the definition of store as given in \cite{montesi:book}, which is the standard way to model the state of processes in choreographies.
For this purpose, we assume a set of values $\Valset$ to be given and we assume that this set can be used to evaluate expressions in $\Eset$.
\begin{definition}
    A \defn{process store} is a (partial) function $\sstore: \Varset \to \Valset$ mapping variables in $\Varset$ to values in $\Valset$.
    We write $\evalin{\sstore}{e}{v}$ to denote that the expression $\ee$ evaluates to the value $v$ in the process store $\sstore$.
    \footnote{We assume that each expression $\ee$ can be evaluated in any process store $\sstore$, that is, for all $\sstore$ and $\ee$ there is a unique $v$ such that $\evalin{\sstore}{e}{v}$ and that can be effectively computed.}
    We denote by $\sstore\lupdate yv$ the process store obtained from $\sstore$ by updating the value of the variable $y$ to $v$, that is, $\sstore\lupdate yv$ is the process store such that $\sstore\lupdate yv (x)=v$ if $x=y$, and $\sstore\lupdate yv (x)=\sstore(x)$ otherwise.

    A \defn{(global) store} is a function $\sStore: \Pset \to (\Varset \to \Valset)$ mapping each process in $\Pset$ to a process store.
    We denote by $\sStore\update pyv$ the global store obtained from $\sStore$ by updating the value of the variable $y$ at process $p$ to $v$, that is, $\sStore\update pyv$ is the global store defined as $\sStore(q)\lupdate yv$ if $q=p$, and as $\sStore(p)$ everywhere else.

    A \defn{stateful choreography} is a pair $\pair\cC\sStore$ where $\cC$ is a choreography and $\sStore$ is a global store.
    A \defn{stateful network} is a pair $\pair\nN\sStore$ where $\nN$ is a network and $\sStore$ is a global store.
    A \defn{stateful process} is a pair $\pair\pP\sStore$ where $\pP$ is a process and $\sstore$ is a process store.
\end{definition}

The \defn{stateful semantics} for stateful choreographies and for stateful networks are  defined by the following rule schemes with the label $\mu$ in the same set $\Glab$ of global labels defined in \Cref{eq:Glab} used by the symbolic semantics of choreographies, and the \defn{state semantics} defined by the rules in \Cref{fig:state-semantics}.
\begin{equation}\label{eq:stateful_semantics}
    \inferrule*[right=stateful-chor]{
        \cC \tradsem{\mu} \cCp
    \qquad
        \sStore \tradsem{\mu} \sStorep
    }{
        \pair\cC\sStore \statefulsem{\mu} \pair{\cCp}{\sStorep}
    }
\hskip6em
     \inferrule*[right=stateful-net]{
        \nN \tradsem{\mu} \nNp
    \qquad
        \sStore \tradsem{\mu} \sStorep
    }{
        \pair\nN\sStore \statefulsem{\mu} \pair{\nNp}{\sStorep}
    }
\end{equation}
This rules allows for the choreography \resp{the network} and the store to evolve together, and it allows us to keep track of the values of the variables at each process as the choreography \resp{the network} evolves.

\begin{figure}
    \def\myskip{\hskip4em}
    \centering
    $\begin{array}{c}
        \inferrule*[right=assign] {  \evalin{\sStore(p)}{e}{v} }
        {\sStore  \tradsem{\TRassign[p]{x}{e}}  \sStore\update p{x}{v}}
    \myskip
        \inferrule*[right=com] { \evalin{\sStore(p)}{e}{v} } 
        { \sStore \tradsem{\TRcom pqx}  \sStore\update qxv }
    \myskip
        \inferrule*[right=no-update] {  } 
        {\sStore \tradsem{\mu} \sStore}
    \\[15pt]
    \end{array}$

    \caption{State semantics, where the label $\mu$ in the rule $\textsc{no-update}$ can be any label in $\Glab\setminus\set{\TRassign[p]{x}{e}, \TRcom pqx}$.}
    \label{fig:state-semantics}
\end{figure}

Observe that, by construction, this modular approach guarantees that bisimilarity under the symbolic semantics guarantees bisimilarity under the stateful semantics for free:
\begin{equation}\label{eq:stateful_bisim}
    \cC \sim \cC' \implies \pair\cC\sStore \sim \pair{\cC'}{\sStore}
    \qquad
    \nN \sim \nN' \implies \pair\nN\sStore \sim \pair{\nN'}{\sStore}
    \qquad
    \cC \sim \nN \implies \pair\cC\sStore \sim \pair{\nN}{\sStore}
\end{equation}
This allows us to easily extend all the results presented in the previous sections to the stateful setting, by simply replacing the symbolic semantics of choreographies and networks with the stateful semantics defined by the rules in \Cref{eq:stateful_semantics}, and by replacing the bisimilarity relation between choreographies and networks with a bisimilarity relation between stateful choreographies and networks.

\section{Related Work}
\label{sec:related}

\paragraph{The merge operator.}
The merge operator was first formulated in \cite{CHY12}.
This has given rise to the now standard problem of how to relate choreographies to their projections after the execution of a conditional: while the choreography commits to a single branch, its
projection may still expose behaviours originating from both branches.
The original solution of \cite{CHY12} addressed this challenge through a syntactic pruning relation. 
In \cite{montesi:book}, the theory of the merge has been reformulated as a join operator over an algebraic structure of process terms, yielding a semilattice whose associated preorder is, in fact, the branching preorder we use in this paper.
These works treated the branching preorder as a syntactic relation, whereas this work provides a purely semantic characterisation of the relation between a choreography and its projection.

\paragraph{Bisimilarity between choreographies and their projections.}
While proving bisimilarity between choreographies and their projections, previous works have often struggled with the treatment of recursion.
In particular, in works like \cite{CM13,PQM25}, the expansion of a choreography-level procedure may require multiple local transitions to expand the procedure call at each process, and these expansions can interleave with other transitions, giving rise to cumbersome syntactic bureacracy in the bisimulation proof.
To avoid this issue, some works sacrifice the fully concurrent nature of procedure invocation, either by introducing a centralised coordinator responsible for procedure entry \cite{DH12} or by requiring a global synchronisation among all participants whenever a procedure is entered, regardless of whether they participate in the procedure or not \cite{HG22}.
An alternative solution proposed in \cite{montesi:book} enriches choreographies with runtime terms that explicitly track the participants that have already entered a procedure and those that have not.
This approach achieves a strong bisimilarity result between a choreography and its projection, but at the cost of extending the choreography language with additional runtime constructs.

In contrast, our framework achieves a strong bisimilarity result between a choreography and its projection while retaining a fully distributed treatment of procedure calls and without extending the choreography language with additional runtime constructs.

\paragraph{The merge operator in multiparty session types.}
The merge operator has also become a standard component of the theory of multiparty session types, where choreographies are used as behavioural specifications rather than executable programs. 
Following its introduction in \cite{CHY12}, recent developments distinguish between the traditional merge operator, so-called \emph{full merge operator}, and the so-called \emph{plain merge operator}, which is defined as an equivalence relation on process terms \cite{HYK26}.
The semantic perspective developed in this paper suggests a new way of understanding the behavioural consequences of merge, potentially contributing to a broader understanding of branching compatibility beyond the setting considered here.

\section{Conclusion}\label{sec:conclusion}

In this paper, we have designed a new semantics for choreographies based on GSOS rules whose conditions are defined locally at each participant, and we have shown that this semantics is equivalent to the traditional semantics of choreographies.
Then, using this semantics we have defined a new relation between choreographies and networks, called the branchin preorder, which is a strengthening of the traditional relation between choreographies and their projections.
Using this new relation, we have shown that the endpoint projection of a choreography is more branching than the original choreography, and we have used this result to give a new proof of the EPP theorem for choreographies, which states that the endpoint projection of a choreography is bisimilar to the original choreography.

This work opens up several avenues for future research, including the study of the branching preorder in other settings, such as in the presence of asynchronous communication or in the context of multiparty session types, and the investigation of the connections between the branching preorder and other behavioural preorders, such as the may and must preorders \cite{den:hen:testing,den:hen:CCStau,hennessy1988algebraic,ber:hen:testing}.

\bibliography{biblio}

@article{P04a,
  author       = {Gordon D. Plotkin},
  title        = {A structural approach to operational semantics},
  journal      = {J. Log. Algebraic Methods Program.},
  volume       = {60-61},
  pages        = {17--139},
  year         = {2004},
  bibsource    = {dblp computer science bibliography, https://dblp.org}
}

@inproceedings{PGSN22,
  author       = {Johannes {\AA}man Pohjola and
                  Alejandro G{\'{o}}mez{-}Londo{\~{n}}o and
                  James Shaker and
                  Michael Norrish},
  editor       = {June Andronick and
                  Leonardo de Moura},
  title        = {Kalas: {A} Verified, End-To-End Compiler for a Choreographic Language},
  booktitle    = {13th International Conference on Interactive Theorem Proving, {ITP}
                  2022, Haifa, Israel, August 7-10, 2022},
  series       = {LIPIcs},
  volume       = {237},
  pages        = {27:1--27:18},
  publisher    = {Schloss Dagstuhl - Leibniz-Zentrum f{\"{u}}r Informatik},
  year         = {2022},
  url          = {https://doi.org/10.4230/LIPIcs.ITP.2022.27},
  doi          = {10.4230/LIPICS.ITP.2022.27},
  bibsource    = {dblp computer science bibliography, https://dblp.org}
}

@article{SHC25,
  author       = {Ashley Samuelson and
                  Andrew K. Hirsch and
                  Ethan Cecchetti},
  title        = {Choreographic Quick Changes: First-Class Location (Set) Polymorphism},
  journal      = {Proc. {ACM} Program. Lang.},
  volume       = {9},
  number       = {{OOPSLA2}},
  pages        = {1783--1808},
  year         = {2025},
  url          = {https://doi.org/10.1145/3763114},
  doi          = {10.1145/3763114},
  bibsource    = {dblp computer science bibliography, https://dblp.org}
}

@inproceedings{DH12,
  author       = {Romain Demangeon and
                  Kohei Honda},
  editor       = {Maciej Koutny and
                  Irek Ulidowski},
  title        = {Nested Protocols in Session Types},
  booktitle    = {{CONCUR} 2012 - Concurrency Theory - 23rd International Conference,
                  {CONCUR} 2012, Newcastle upon Tyne, UK, September 4-7, 2012. Proceedings},
  series       = {Lecture Notes in Computer Science},
  volume       = {7454},
  pages        = {272--286},
  publisher    = {Springer},
  year         = {2012},
  url          = {https://doi.org/10.1007/978-3-642-32940-1\_20},
  doi          = {10.1007/978-3-642-32940-1\_20},
  bibsource    = {dblp computer science bibliography, https://dblp.org}
}

@article{PQM25,
  author       = {Dan Plyukhin and
                  Xueying Qin and
                  Fabrizio Montesi},
  title        = {Relax! The Semilenient Core of Choreographic Programming (Functional
                  Pearl)},
  journal      = {Proc. {ACM} Program. Lang.},
  volume       = {9},
  number       = {{ICFP}},
  pages        = {947--973},
  year         = {2025},
  url          = {https://doi.org/10.1145/3747538},
  doi          = {10.1145/3747538},
  bibsource    = {dblp computer science bibliography, https://dblp.org}
}

@article{HYK26,
  author       = {Ping Hou and
                  Nobuko Yoshida and
                  Iona Kuhn},
  title        = {Less is more revisited: Association with global protocols and multiparty
                  sessions},
  journal      = {Theor. Comput. Sci.},
  volume       = {1076},
  pages        = {115873},
  year         = {2026},
  url          = {https://doi.org/10.1016/j.tcs.2026.115873},
  doi          = {10.1016/J.TCS.2026.115873},
  bibsource    = {dblp computer science bibliography, https://dblp.org}
}

@article{CHY12,
  author       = {Marco Carbone and
                  Kohei Honda and
                  Nobuko Yoshida},
  title        = {Structured Communication-Centered Programming for Web Services},
  journal      = {{ACM} Trans. Program. Lang. Syst.},
  volume       = {34},
  number       = {2},
  pages        = {8:1--8:78},
  year         = {2012},
  url          = {https://doi.org/10.1145/2220365.2220367},
  doi          = {10.1145/2220365.2220367},
  bibsource    = {dblp computer science bibliography, https://dblp.org}
}

@inproceedings{CM13,
  author       = {Marco Carbone and
                  Fabrizio Montesi},
  editor       = {Roberto Giacobazzi and
                  Radhia Cousot},
  title        = {Deadlock-freedom-by-design: multiparty asynchronous global programming},
  booktitle    = {The 40th Annual {ACM} {SIGPLAN-SIGACT} Symposium on Principles of
                  Programming Languages, {POPL} '13, Rome, Italy - January 23 - 25,
                  2013},
  pages        = {263--274},
  publisher    = {{ACM}},
  year         = {2013},
  url          = {https://doi.org/10.1145/2429069.2429101},
  doi          = {10.1145/2429069.2429101},
  bibsource    = {dblp computer science bibliography, https://dblp.org}
}

@PhdThesis{M13:phd,
  author = "Fabrizio Montesi",
  title = "Choreographic {P}rogramming",
  school = "IT University of Copenhagen",
  type = "Ph.{D}. Thesis",
  year = 2013,
  note = {\url{https://www.fabriziomontesi.com/files/choreographic-programming.pdf}}
}

@article{den:hen:testing,
	title = {Testing equivalences for processes},
	journal = {Theoretical Computer Science},
	volume = {34},
	number = {1},
	pages = {83-133},
	year = {1984},
	issn = {0304-3975},
	doi = {https://doi.org/10.1016/0304-3975(84)90113-0},
	url = {https://www.sciencedirect.com/science/article/pii/0304397584901130},
	author = {R. {De Nicola} and M.C.B. Hennessy}
}

@book{hennessy1988algebraic,
	author = {Hennessy, Matthew},
	title = {Algebraic theory of processes},
	year = {1988},
	isbn = {0262081717},
	publisher = {MIT Press},
	address = {Cambridge, MA, USA}
}

@InProceedings{den:hen:CCStau,
	author="De Nicola, Rocco
	and Hennessy, Matthew",
	editor="Ehrig, Hartmut
	and Kowalski, Robert
	and Levi, Giorgio
	and Montanari, Ugo",
	title="CCS without $\tau$'s",
	booktitle="TAPSOFT '87",
	year="1987",
	publisher="Springer Berlin Heidelberg",
	address="Berlin, Heidelberg",
	pages="138--152",
	isbn="978-3-540-47746-4"
}

@article{ber:hen:testing,
	TITLE = {{Mutually Testing Processes}},
	AUTHOR = {Giovanni Bernardi and Matthew Hennessy},
	URL = {https://lmcs.episciences.org/776},
	DOI = {10.2168/LMCS-11(2:1)2015},
	JOURNAL = {{Logical Methods in Computer Science}},
	VOLUME = {{Volume 11, Issue 2}},
	YEAR = {2015},
	MONTH = Apr,
}

@InProceedings{acc:man:mon:ESOP25,
author="Acclavio, Matteo
and Manara, Giulia
and Montesi, Fabrizio",
editor="Vafeiadis, Viktor",
title="Formulas as Processes, Deadlock-Freedom as Choreographies",
booktitle="Programming Languages and Systems",
year="2025",
publisher="Springer Nature Switzerland",
address="Cham",
pages="23--55",
isbn="978-3-031-91118-7"
}

@misc{acc:mon:per:OPDL,
	title={On Propositional Dynamic Logic and Concurrency},
	author={Matteo Acclavio and Fabrizio Montesi and Marco Peressotti},
	year={2024},
	eprint={2403.18508},
	archivePrefix={arXiv},
	primaryClass={cs.LO}
}

@book{montesi:book,
	author={Montesi, Fabrizio},
	title={Introduction to Choreographies},
	place={Cambridge},
	doi={10.1017/9781108981491},
	publisher={Cambridge University Press},
	year={2023}
}

@PhdThesis{montesi:phd,
	author = "Fabrizio Montesi",
	title = "Choreographic {P}rogramming",
	school = "IT University of Copenhagen",
	type = "Ph.{D}. Thesis",
	year = 2013,
	note = {\url{https://www.fabriziomontesi.com/files/choreographic-programming.pdf}}
}

@article{HG22,
	author       = {Andrew K. Hirsch and
	Deepak Garg},
	title        = {Pirouette: higher-order typed functional choreographies},
	journal      = {Proc. {ACM} Program. Lang.},
	volume       = {6},
	number       = {{POPL}},
	pages        = {1--27},
	year         = {2022},
	url          = {https://doi.org/10.1145/3498684},
	doi          = {10.1145/3498684},
	bibsource    = {dblp computer science bibliography, https://dblp.org}
}

@book{SW01,
	author       = {Davide Sangiorgi and
	David Walker},
	title        = {The Pi-Calculus - a theory of mobile processes},
	publisher    = {Cambridge University Press},
	year         = {2001},
	isbn         = {978-0-521-78177-0},
	bibsource    = {dblp computer science bibliography, https://dblp.org}
}

@book{S11,
	title={Introduction to bisimulation and coinduction},
	author={Sangiorgi, Davide},
	year={2011},
	publisher={Cambridge University Press}
}

@book{M80,
	author       = {Robin Milner},
	title        = {A Calculus of Communicating Systems},
	series       = {Lecture Notes in Computer Science},
	volume       = {92},
	publisher    = {Springer},
	year         = {1980},
	url          = {https://doi.org/10.1007/3-540-10235-3},
	doi          = {10.1007/3-540-10235-3},
	isbn         = {3-540-10235-3},
	bibsource    = {dblp computer science bibliography, https://dblp.org}
}

@article{CMP23,
	author = {Cruz{-}Filipe, Lu{\'{\i}}s and Montesi, Fabrizio and Peressotti, Marco},
	title = {A Formal Theory of Choreographic Programming},
	journal = {Journal of Automated Reasoning},
	volume = {67},
	number = {21},
	pages = {1--34},
	year = {2023},
	issn = {1573-0670},
	url = {https://doi.org/10.1007/s10817-023-09665-3},
	doi = {10.1007/s10817-023-09665-3}
}

\section{Acknowledgments}
Co-funded by the European Union (ERC, CHORDS, 101124225). 
Views and opinions expressed are however those of the authors only and do not necessarily reflect those of the European Union or the European Research Council. Neither the European Union nor the granting authority can be held responsible for them.

\clearpage

\appendix
\section{Detailed Proofs of Lemmas}\label{app:proofs}

\lemmaMBbasic*
\begin{proof}
    We prove each statement separately.    
    \begin{enumerate}
        \item  By induction on the derivation of $\nN \netsem{\mu} \nNp$, doing case analysis on the last rule of the derivation.
        \begin{itemize}
            \item If the last rule is a \parleft, then $\nN = \nN[1] \ppar \nN[2]$ and $\nNp = \nNp[1] \ppar \nN[2]$ with $\nN[1] \netsem{\mu} \nNp[1]$. 
            If $\pid r \in \dom{\nN[2]}$, then we deduce by \Cref{lem:mb_net}.\ref{lem:mb_net2} and  \Cref{lem:mb_net}.\ref{lem:mb_net3} that $\cC \mbbn{r} \nN[2]$. We conclude by induction hypothesis and \Cref{lem:mb_net}.\ref{lem:mb_net2}.
            \item If the last rule is a \parright, then the proof is symmetric to the previous case. 
            \item Otherwise, $\nN = \pdo[p]{\pP}$ for some process term $\pP$. We conclude observing that there are no transitions at $\pid r$ for both $\cC$ and $\nNp$.
        \end{itemize}
        %

    
        \item We show both implications separately, by letting $\cC=\cifte pe{\cC[1]}{\cC[2]} $ and assuming $\cC \mbbn[n+1]{r} \nN$.
        \begin{itemize}
            \item[($\Rightarrow$)]
            We proceed by induction on $n$:
            \begin{itemize}
                \item Case base:
                By \Cref{def:mbb_strat}.\ref{def:mbb_strat0} if $\pid p \in \pnof{\cC}$, then $\nN(\pid p) \neq \pnil$.
                We conclude observing that by \Cref{cor:well_guardedness} $\pid r \in \pnof{\cC[1]}$ if and only if $\pid r \in \pnof{{\cC[2]}}$.
                Therefore, $\pid r \in \pnof{\cC[1]}$ if and only if $\pid r \in \pnof{\cC}$.

                \item Inductive step:
                We only show that $\cC[1] \mbbn[n+1]{r} \nN$, since $\cC[2] \mbbn[n+1]{r} \nN$ the prove is similar.
                We check the conditions of the definition of $\mbbn[n+1]{r}$:
                \begin{itemize}
                    \item[$\bullet$] \mbrel:
                    Since $\mbbn{r}$ is a decreasing relation (see \Cref{lem:decreasing}) $\cC\mbbn[n]{r} \nN$. We conclude by induction hypothesis and \Cref{def:mbb_strat}.\ref{def:mbb_strat0} applied to $\cC[1] \mbbn[n]{r} \nN$.

                    \item[$\bullet$] \mbGtoL:
                    If $\cC[1] \localsem[r]{\mu} \cCp[1]$, then by \Cref{lem:mb_wf_merge} we have the either of the following cases:
                    \begin{itemize}
                        \item $\cC[2] \localsem[r]{\mu} \cCp[2]$: by \delaycond[r], $\cC \localsem[r]{\mu} \cifte pe{\cCp[1]}{\cCp[2]}$. Hence, there exists $\nNp$ such that $\nN \netsem{\mu} \nNp$ and $\cifte pe{\cCp[1]}{\cCp[2]} \mbbn[n]{r} \nNp$. We conclude by induction hypothesis.

                        \item $\mu = \PRlrecv[$\albl$]{q}{r}$ and $\cC[j] \localsem[r]{\PRlrecv[$\overline{\albl}$]{q}{r}} \cCp[j]$:
                        by \delaypredictl[r] $\cC \localsem[r]{\mu} \cCp[1]$. We conclude directly from $\cC \mbbn[n+1]{r} \nN$ and by \Cref{def:mbb_strat}.
                    \end{itemize}
                
                    \item[$\bullet$] \mbLtoG:
                    If $\nN \netsem{\mu} \nNp$ with $\pnof{\mu} = \set{\pid r}$, then, since $\cC\mbbn[n+1]{r} \nN$, 
                    by \Cref{def:mbb_strat}.\ref{def:mbb_strat2} there are two possibilities for $\cC$:
                    \begin{enumerate}
                        \item $\cC$ can perform a transition at the same label $\mu$.
                        We proceed by induction on the derivation of this transition, doing case analysis on the last rule of the derivation.
                        \begin{itemize}
                            \item If the last rule is a \delaycond[r], then $\cCp = \cifte pe{\cCp[1]}{\cCp[2]} \mbbn[n]{r} \nN$ and we conclude by induction hypothesis. 
                            \item If the last rule is a \delaypredictl[r] we coclude directly by definition of $\mbbn[n+1]{r}$.
                            \item If the last rule is a \delaypredictl[l] we deduce by definition of local semantics (see \Cref{fig:new-semantics}) that $\mu = \PRlrecv[${\albl}$]{q}{r}$ and that  $\cC[1] \localsem[r]{\PRlrecv[$\overline{\albl}$]{q}{r}} \cCp[1]$. The we can derive $\cC \localsem[r]{\PRlrecv[$\overline{\albl}$]{q}{r}} \cCp[1]$.
                            We conclude by \Cref{def:mbb_strat}.\ref{def:mbb_strat1} applied to $\cC \mbbn[n+1]{r} \nN$.
                        \end{itemize}

                        \item $\cC$ cannot perform a transition at $\mu$, $\mu= \PRlrecv[${\albl}$]{q}{r}$, and $\cC \localsem[r]{\PRlrecv[$\overline{\albl}$]{q}{r}} \cCpp$.
                        By \Cref{lem:mb_wf_merge} and definition of local semantics we deduce that a transition at $\mu= \PRlrecv[${\albl}$]{q}{r}$ must be undefined for both $\cC[1]$ and $\cC[2]$
                        and a trasition at $\mu= \PRlrecv[$\overline{\albl}$]{q}{r}$ must be defined for both $\cC[1]$ and $\cC[2]$. 
                        Hence $\cCpp = \cifte pe{\cCp[1]}{\cCp[2]}$ with $\cC[i] \localsem[r]{\PRlrecv[$\overline{\albl}$]{q}{r}} \cCp[i]$ for $i \in\set{ 1,2}$. 
                        Moreover, by \Cref{def:mbb_strat}.\ref{def:mbb_strat2} there is $\nNpp$ such that $\nN \netsem{\PRlrecv[$\overline{\albl}$]{q}{r}} \nNpp$ and $\cCpp \mbbn[n]{r} \nNpp$.
                        We conclude by induction hypothesis. 
                    \end{enumerate}
                \end{itemize}
            \end{itemize}

            \item[($\Leftarrow$)]
            By induction on $n$.
            For $n=0$ the result is trivial. For $n>0$,
            we prove that if 
            $\cC[i] \mbbn[n+1]{r} \nN$ for both $i\in\set{1,2}$, then 
            $\cC \mbbn[n+1]{r} \nN$.
            We check that the conditions of \Cref{def:mbb_strat} holds for $\cC \mbbn[n+1]{r} \nN$:

            \begin{itemize}
                \item \mbrel: 
                follows directly by $\cC[i] \mbbn[n+1]{r} \nN$ for $i=1,2$. In particular, we observe that $\pid p \in \pnof{\cifte pe{\cC[1]}{\cC[2]}}$ if and only if
                $\pid p \in \pnof{\cC[1]}$ or $\pid p \in \pnof{\cC[2]}$.

                \item \mbGtoL:
                suppose
                $\cC \localsem[r]{\mu} \cCp$.
                If the last rule is a \delaycond[r] we conclude by induction hypothesis. 
                If the last rule is a \delaypredictl[r] (resp. \delaypredictr[r]) we conclude by hypothesis on $\cC[1]$ (resp. $\cC[2]$).

                \item \mbLtoG: 
                suppose
                $\nN \netsem{\mu} \nNp$
                with
                $\pnof{\mu} = \set{\pid r}$.

                From the fact that
                $\cC[i] \mbbn[n+1]{r} \nN$
                for $i=\set{1,2}$,
                we deduce that
                $\cC[i] \localsem[r]{\mu_i} \cCp[i]$
                with
                $\mu_i \in \set{\mu,\PRlrecv[$\overline{\albl}$]{q}{r}}$.
                In particular $\mu_i = \PRlrecv[$\overline{\albl}$]{q}{r}$ {whenever} $\mu= \PRlrecv[${\albl}$]{q}{r}$ and there is no transition $\mu$ for $\cC[i]$ at $\pid r$.

                If $\mu_1 = \mu_2$, then
                $\cifte pe{\cC[1]}{\cC[2]} \localsem[r]{\mu} \cifte pe{\cCp[1]}{\cCp[2]}$.
                and we conclude by the induction hypothesis.

                Otherwise
                $\cC \localsem[r]{\mu} \cCp[i]$
                for some $i \in \set{1,2}$,
                and we conclude using the hypothesis that
                $\cC[i] \mbbn[n+1]{r} \nN$.
            \end{itemize}

        \end{itemize}

        \item It is sufficient to observe that when $\pid r \notin \cI$ then $\cI \seq \cC$ can make a transition at $\pid r$ if and only if $\cC$ can.
    \end{enumerate}
\end{proof}

\lemmaUgly*
\begin{proof}
    Items~\ref{lem:mb_net4}, \ref{lem:mb_net5}, and~\ref{lem:mb_net6} follow directly from 
    \Cref{def:mbb_strat} and the definitions of network semantics and local symbolic semantics.
    For the remaining:
    
    \begin{enumerate}
    \item We show both implications separately. 
    \begin{itemize}
        \item [($\Rightarrow$)]
        We observe that since $\cI$ and $\palpha$ are compatible the only possible transition $\mu$ at $\pid r$ is the one such that 
        $\cI \seq \cC \localsem[r]{\mu} \cC$ and
        $\pdo[r]{\palpha \pseq \pR} \netsem{\mu} \pdo[r]{\pR}$. By definition of $\mbbn[n+1]{r}$, this implies $\cC \mbbn[n]{r} \pdo[r]{ \pR}$.

        \item [($\Leftarrow$)]
        We show that $\cI \seq \cC \mbbn[n+1]{r} \pdo[r]{\palpha \pseq \pR}$ by checking the conditions of the definition of 
        $\mbbn[n+1]{r}$.  

        \begin{itemize}
            \item \mbrel: trivial;
            \item \mbGtoL: For what observed above the only possible transition for $\cI \seq \cC$  at $\pid r$ is $\mu$ such that $\cI \seq \cC \localsem[r]{\mu} \cC$.
            We conclude observing that $\pdo[r]{\palpha \pseq \pR} \netsem{\mu} \pdo[r]{\pR}$ and that $\cC \mbbn[n]{r} \pdo[r]{\pR}$ by hypothesis.

            \item \mbLtoG:  We observe that only possible transition $\mu$, with $\pid r \in \pnof{\mu}$ is the one such that $\cI \seq \cC \localsem[r]{\mu} \cC$. We conclude reasoning as above.
        \end{itemize}
    \end{itemize}

    \item We show both implications separately. 
        \begin{itemize}
            \item [($\Rightarrow$)]
            We observe that the only possible transitions at the process $\pid r$ for the choreography $\ccsel pr\albl \seq \cC$ and the network $\pdo[r]{\plabel[\pR] r\albl\Lset}$ have label $\mu = \PRlrecv[\albl]{p}{r}$ with $\albl \in \Lset$. 
            By definition of $\mbbn[n+1]{r}$, this implies that $\cC \mbbn[n]{r} \pdo[r]{\pR[\albl]}$ for $\albl \in \Lset$. 

            \item [($\Leftarrow$)]
            We show that $\ccsel pr\albl \seq \cC \mbbn[n+1]{r} \pdo[r]{\plabel[\pR] r\albl\Lset}$ by checking the conditions of the definition of $\mbbn[n+1]{r}$.  

            \begin{itemize}
                \item \mbrel: trivial;
                \item \mbGtoL: For what observed above the only possible transitions at $\pid r$ for $\ccsel pr\albl \seq \cC $ are $\mu = \PRlrecv[$\albl$]{p}{r}$ for $\albl \in \Lset$. We conclude observing that $\pdo[r]{\plabel[\pR] r\albl\Lset} \netsem{\mu} \pdo[r]{\pR[\albl]}$  and that $\cC \mbbn[n]{r} \pdo[r]{\pR[\albl]}$ by hypothesis.

                \item \mbLtoG: We observe that only possible transitions $\mu$ for $\pdo[r]{\plabel[\pR] r\albl\Lset}$, are $\mu = \PRlrecv[$\albl$]{p}{r}$ for $\albl \in \Lset$. 
                We conclude reasoning as above.
            \end{itemize}
        \end{itemize}

    \item We show both implications separately. 
    \begin{itemize}
        \item[($\Rightarrow$)]
        Derives directly by definition of local semantics (see \Cref{fig:local-semantics}), network semantics and of $\mbbn[n+1]{p}$.

        \item[($\Leftarrow$)]
        We check the conditions of the definition of $\mbbn[n+1]{p}$. 
        \begin{itemize}
            \item \mbrel: trivial;
            \item \mbGtoL: If $\cifte pe{\cC[1]}{\cC[2]} \localsem[p]{\mu} \cCp$ then either $\mu = {\TRthen[p]{e}} $ and $\cCp = \cC[1]$ or $\mu ={\TRelse[p]{e}}$ and $\cCp = \cC[2]$. We conclude observing that $\pdo[p]{\pifte e {\pP[1]}{\pP[2]}} \netsem{\TRthen[p]{e}} \pdo[p]{\pP[1]}$ (resp. $\pdo[p]{\pifte e{\pP[1]}{\pP[2]}} \netsem{\TRelse[p]{e}} \pdo[p]{\pP[2]}$) and by hypothesis $\cC[1] \mbbn{p} \pdo[p]{\pP[1]}$ (resp. $\cC[2] \mbbn{r} \pdo[p]{P_2}$).
            \item \mbLtoG: If  $\pdo[p]{\pifte e {\pP[1]}{\pP[2]}} \netsem{\mu} \nNp$ then  either $\mu = {\TRthen[p]{e}} $ and $\nNp = \pdo[p]{\pP[1]}$ or $\mu ={\TRelse[p]{e}}$ and $\nNp = \pdo[p]{\pP[2]}$. We conclude by reasoning as above.
        \end{itemize}
    \end{itemize}
\end{enumerate}
\end{proof}

\thmStratifiedSimu*
\begin{proof}
    For $n=0$, the result is trivial.
    For $n > 0$, we show that $\mbbn[n+1]{}$ meets the conditions of \Cref{def:mbb_strat}.\ref{mbb_strat:Step}. 
    Note that we prove the \mbLtoG condition before the \mbGtoL condition, since we prioritize the details of \mbLtoG, which are more intricate, to avoid having to repeat the same case analysis twice.
    \begin{itemize}
        \item \mbrel: trivial by definition of $\mbbn[0]{}$;
        \item \mbLtoG: by case analysis on $\mu$. 
        
        We first show that if $\nN \netsem{\mu} \nNp$ then $\cC \tradsem{\mu} \cCp$ for some $\cCp$ by case analysis on $\mu$, reasoning on the aggregate semantics of choreographies and on the network semantics, and by applying \Cref{thm:agg-trad} to conclude
        \begin{itemize}
            \item 
            if $\mu \in \set{\TRassign[p] xe, \TRthen e, \TRelse e}$, 
            then $\cC \localsem[p]{\mu} \cCp$  by the \mbLtoG condition and by definition of aggregate semantics of choreographies.

            \item 
            if $\mu = \TRcom pqx$, 
            then, by the definition of the network semantics, 
            we must have $\nN=\nN[1]\ppar \nN_2$ 
            for some networks $\nN[1]$ and $\nN[2]$ such that 
            $\nN[1] \netsem{\PRsend pq} \nNp[1]$ and $\nN[2] \netsem{\PRrecv pq} \nNp[2]$.
            Since 
            $\cC \mbbn[n+1]{p} \nN$ and $\cC \mbbn[n+1]{q} \nN$, we must also have $\cC \localsem[p]{\PRsend pq} \cCp$ and $\cC \localsem[q]{\PRrecv pq} \cCpp$. 
            We conclude by showing that $\cCp = \cCpp$, reasoning by induction on the derivation of the transition $\cC \localsem[p]{\PRsend pq} \cCp$ in the local semantics of choreographies, doing case analysis of the label of the bottom-most rule used in the derivation:
            \begin{itemize}
                \item case \csend[p] is immediate, since it implies $\cC = \cgencom \seq \cCp$;
                \item cases \delaypredictl[p] and \delaypredictr[p] are excluded by definition of local semantics (see \Cref{fig:local-semantics});
                \item case \delaycond[p] requires that $\cC = \cifte re{\cC[1]}{\cC[2]}$ and 
                $\cCp = \cifte re{\cCp[1]}{\cCp[2]}$. In this case, by \Cref{lem:mb_wf_merge}, we must have $\cC \localsem[q]{\PRrecv pq} \cifte re{\cCpp[1]}{\cCpp[2]}$. We conclude by induction hypothesis;
                \item cases \cdelay[p] and \ccall[p] follow immediately by induction hypothesis as well.
            \end{itemize}
            
            \item 
            if $\mu = \TRsel pq$, we conclude similarly to the previous case, considering the local semantics of choreographies with labels $\PRlsend pq$ and $\PRlrecv pq$ instead of $\PRsend pq$ and $\PRrecv pq$.
        \end{itemize}
        
        Then, in order to conclude, it suffices to show that $\cCp \mbbn[n]{} \nNp$ for all $\pid r \in \Pset$.
        Notice that if $\pid r \notin \pnof{\cC}$, then  the statement holds trivially; therefore we assume $\pid r \in \pnof{\cC}$ in the following.
        We distinguish two cases:
        \begin{itemize}
            \item if $\pid r \in \pnof{\mu}$, then the required $\mbbn[n]{r}$ relation follows directly from \Cref{def:more_branch_cn}, and we conclude by \Cref{lem:mb}.\ref{lem:mb1}.
            Specifically, if $\mu \in \set{ \TRcom pqx, \TRsel pq}$, then $\cCp \mbbn[n]{p} \nNp[1] \ppar \nN[2]$ and 
            $\cCp \mbbn[n]{q} \nN[1] \ppar \nNp[2]$, and we can apply \Cref{lem:mb}.\ref{lem:mb1} twice;

            \item otherwise, if $\pid r \notin \pnof{\mu}$, then $\cCp \mbbn[n]{r} \nN$ and we conclude by applying \Cref{lem:mb}.\ref{lem:mb1} multiple times, proceeding by induction on the derivation of $\cC \tradsem{\mu} \cCp$ in the traditional semantics, doing case analysis of the label of the bottom-most rule used in the derivation:
                \begin{itemize}
                    \item cases \tassign, \tcom, and \tsel : we must have $\cC = \cI \seq\cCp$ for an adequate instruction $\cI$. We conclude by induction hypothesis using \Cref{lem:mb}.\ref{lem:mb3} and \Cref{lem:decreasing}.

                    \item case \tdelay : we must have $\cC = \cI \seq \cC[1]$ with $\pnof{\mu} \cap \pnof{\cI}=\emptyset$.
                    We distinguish two sub-cases:
                    \begin{itemize}
                        \item if $\pid r \notin \pnof{\cI}$, then by \Cref{lem:mb}.\ref{lem:mb3} we have $\cC[1] \mbbn[n+1]{r} \nN$; 
                        We conclude by applying induction hypothesis and \Cref{lem:mb}.\ref{lem:mb3};

                        \item if $\pid r \in \pnof{\cI}$, then by \Cref{lem:mb_net}.\ref{lem:mb_net3} 
                        and \Cref{lem:ugly}.\ref{lem:mb_net4} we have that 
                        either $\nN(\pid r) = \palpha \pseq \pR$ with $\alpha$ compatible with $\cI$ and $\pid r \notin \pnof{\alpha}$,
                        or $\nN(\pid r)= \plabel[\pR] r\albl\Lset$.
                        In the first case, we conclude by applying the left-to-right implication in \Cref{lem:ugly}.\ref{lem:mb4}, induction hypothesis, and the right-to-left implication in \Cref{lem:ugly}.\ref{lem:mb4} right-to-left.
                        In the second case, we conclude by applying the left-to-right implication in \Cref{lem:ugly}.\ref{lem:mb4bis}, induction hypothesis, then the right-to-left implication in \Cref{lem:ugly}.\ref{lem:mb4bis} and \Cref{lem:mb_net}.\ref{lem:mb_net3}.
                    \end{itemize}

                \item case \tdelaycond : we must have $\cC = \cifte qe{\cC[1]}{\cC[2]}$, and we distinguish two sub-cases:
                \begin{itemize}
                    
                    \item 
                    if $\pid r \neq \pid q$, then we have that
                    $\cC[i] \mbbn[n+1]{r} \nN$ for both $i\in\set{1,2}$ by \Cref{lem:mb}.\ref{lem:mb2}.
                    In this case we conclude by applying induction hypothesis and \Cref{lem:mb}.\ref{lem:mb2};
                    
                    \item 
                    if $\pid r = \pid q$, then we have
                    ${\nN(\pid q)}= {\pifte e {\pQ_1}{\pQ_2}}$ by \Cref{lem:mb_net}.\ref{lem:mb_net3}  and \Cref{lem:ugly}.\ref{lem:prefix2}.
                    By \Cref{lem:ugly}.\ref{lem:mb5} we have that $\cC[i] \mbbn[n]{q} \pdo[q]{\pQ_i}$ for $i\in\set{1,2}$,
                    and by induction hypothesis: $\cCp[i] \mbbn[n-1]{q} \pdo[q]{\pQ_i}$ for $i\in\set{1,2}$.
                    By \Cref{lem:ugly}.\ref{lem:mb5} we obtain $\cifte qe{\cCp[1]}{\cCp[2]} \mbbn[n]{q} \pdo[q]{\nN(\pid q)}$.
                    We conclude by \Cref{lem:mb_net}.\ref{lem:mb_net3}.
                \end{itemize}

                \item case \tthen only holds if $\mu = \TRthen e$: we must have $\cC = \cifte pe{\cCp}{\cC[2]}$. Then we can apply \Cref{lem:mb}.\ref{lem:mb2}, $\cC[1] \mbbn[n+1]{r} \nN$ and conclude by \Cref{lem:decreasing};

                \item case \tcall : we must have $\cX = \cC[1] \in \procC$ and, by \Cref{lem:ugly}.\ref{lem:mb_net6}, $\cC[1] \mbbn[n+1]{} \nN$. We conclude by induction hypothesis on $\cC[1]$.
            \end{itemize}

        \end{itemize}

        \item \mbGtoL: as in the previous case, we reason by case analysis on $\mu$:
        \begin{itemize}
            \item If $\mu \in \set{\TRassign[p] xe, \TRthen e, \TRelse e}$, then if follows from \Cref{def:more_branch_cn}.\ref{def:more_branch_cn1} that $\nN \netsem{\mu} \nNp$.

            \item  If $\mu = \TRcom pqx$, then we consider an atomic network $\nN[1]$ such that $\nN[1](\pid p) = \nN(\pid p)$ (and $\nN[1](\pid r)$ is undefined for $\pid r \neq \pid p$), and a network $\nN[2]$ such that $\nN[1]\ppar \nN[2]=\nN$.
            Since $\cC \mbbn[n+1]{p} \nN$ and $\cC \mbbn[n+1]{q} \nN$, it follows from \Cref{lem:mb_net}.\ref{lem:mb_net2} and \Cref{lem:mb_net}.\ref{lem:mb_net3} that $\nN[1] \netsem{\PRsend pq} \nNp[1]$ and $\nN[2] \netsem{\PRrecv pq} \nNp[2]$.
            We conclude by definition of network semantics.

            \item  If $\mu =\TRsel pq$, then we conclude similarly to the previous case, considering the local semantics of choreographies with labels $\PRlsend pq$ and $\PRlrecv pq$ instead of $\PRsend pq$ and $\PRrecv pq$.

        \end{itemize}
        Once $\nN \netsem{\mu} \nNp$ is established, $\cCp \mbbn[n]{} \nNp$ follows 
        by the same argument as for the \mbLtoG condition.

    \end{itemize}
\end{proof}

\end{document}